\documentclass[a4paper,11pt,reqno]{amsart}

\usepackage{amsmath, amsfonts, amssymb, amsthm, amscd}
\usepackage{geometry}
\usepackage[T1]{fontenc}			
\usepackage[utf8]{inputenc}
\usepackage[english]{babel}
\usepackage[square,numbers]{natbib}
\usepackage[monochrome]{color}	
\usepackage{graphicx,float}
\usepackage[pdftitle=Risk measures and progressive enlargement of filtrations: a BSDE approach,%
pdfauthor={Alessandro Calvia, Emanuela Rosazza Gianin},%
pdfstartview={XYZ null null 1.5},%
pdftex,hidelinks]{hyperref}

\newcommand{\bbE}{{\ensuremath{\mathbb E}} }
\newcommand{\bbF}{{\ensuremath{\mathbb F}} }
\newcommand{\bbG}{{\ensuremath{\mathbb G}} }
\newcommand{\bbH}{{\ensuremath{\mathbb H}} }

\newcommand{\bbP}{{\ensuremath{\mathbb P}} }
\newcommand{\bbQ}{{\ensuremath{\mathbb Q}} }

\newcommand{\cA}{{\ensuremath{\mathcal A}} }
\newcommand{\cB}{{\ensuremath{\mathcal B}} }

\newcommand{\cD}{{\ensuremath{\mathcal D}} }
\newcommand{\cE}{{\ensuremath{\mathcal E}} }
\newcommand{\cF}{{\ensuremath{\mathcal F}} }
\newcommand{\cG}{{\ensuremath{\mathcal G}} }
\newcommand{\cH}{{\ensuremath{\mathcal H}} }

\newcommand{\cM}{{\ensuremath{\mathcal M}} }
\newcommand{\cN}{{\ensuremath{\mathcal N}} }
\newcommand{\cO}{{\ensuremath{\mathcal O}} }
\newcommand{\cP}{{\ensuremath{\mathcal P}} }
\newcommand{\cQ}{{\ensuremath{\mathcal Q}} }

\newcommand{\cS}{{\ensuremath{\mathcal S}} }

\newcommand{\dd}{{\ensuremath{\mathrm d}} }
\newcommand{\de}{{\ensuremath{\mathrm e}} }

\newcommand{\dB}{{\ensuremath{\mathrm B}} }

\newcommand{\dL}{{\ensuremath{\mathrm L}} }

\newcommand{\R}{\mathbb{R}}

\newcommand{\N}{\mathbb{N}}

\newcommand{\prog}{\ensuremath{\mathrm{Prog}}}

\newcommand{\norm}[1]{\lVert #1 \rVert}

\newcommand{\ind}{\ensuremath{\mathbf{1}}}

\DeclareMathOperator*{\esssup}{\mathrm{ess\,sup}}

\newcommand{\coloneqq}{\mathrel{\mathop:}=}
\newcommand{\eqqcolon}{=\mathrel{\mathop:}}

\renewcommand{\epsilon}{\varepsilon}

\newcommand{\thetavartheta}
{
\let\temp\theta
\let\theta\vartheta
\let\vartheta\temp
}

\newcommand{\phivarphi}
{
\let\temp\phi
\let\phi\varphi
\let\varphi\temp
}

\theoremstyle{plain}
\newtheorem{theorem}{Theorem}[section]
\newtheorem*{theorem*}{Theorem}
\newtheorem{lemma}[theorem]{Lemma}
\newtheorem*{lemma*}{Lemma}
\newtheorem{proposition}[theorem]{Proposition}
\newtheorem*{proposition*}{Proposition}

\theoremstyle{definition}
\newtheorem{definition}{Definition}[section]
\newtheorem{assumption}{Assumption}[section]
\newtheorem*{assumption*}{Assumption}

\newtheorem*{hypothesis*}{Hypothesis}
\newtheorem{example}{Example}[section]

\theoremstyle{remark}
\newtheorem{remark}{Remark}[section]
\newtheorem*{remark*}{Remark}

\numberwithin{equation}{section}

\newcommand{\mail}[1]{\href{mailto:#1}{\normalfont\texttt{#1}}}

\makeatletter
\def\@setthanks{\vspace{-\baselineskip}\def\thanks##1{\@par##1\@addpunct.}\thankses}
\makeatother

\title[Risk measures, enlargement of filtration and BSDEs]{Risk measures and progressive enlargement of filtration: a BSDE approach}
\author[A.~Calvia]{Alessandro~Calvia}
\author[E.~Rosazza~Gianin]{Emanuela~Rosazza~Gianin}
\thanks{\noindent A.~Calvia
\\
LUISS "Guido Carli", Department of Economics and Finance, viale Romania 32, 00197, Roma (Italy).
\\
E-mail: \mail{acalvia@luiss.it}.
\medskip
\\
E.~Rosazza~Gianin
\\
University of Milano-Bicocca, Department of Statistics and Quantitative Methods, via Bicocca degli Arcimboldi 8, 20126, Milano (Italy).
\\
E-mail: \mail{emanuela.rosazza1@unimib.it}.
\medskip
\\
This work was done while A. Calvia was Postdoctoral Researcher at the University of Milano-Bicocca, Department of Statistics and Quantitative Methods.
\medskip
}

\begin{document}
	
\begin{abstract}
We consider dynamic risk measures induced by Backward Stochastic Differential Equations {\color{blue} (BSDEs)} in enlargement of filtration setting. On a fixed probability space, we are given a standard Brownian motion and a pair of random variables $(\tau, \zeta) \in (0,+\infty) \times E$, with $E \subset \R^m$, that enlarge the reference filtration, i.e., the one generated by the Brownian motion. These random variables can be interpreted financially as a default time and an associated mark. After introducing a BSDE driven by the Brownian motion and the random measure associated to $(\tau, \zeta)$, we define the dynamic risk measure $(\rho_t)_{t \in [0,T]}$, for a fixed time $T > 0$, induced by its solution. We prove that $(\rho_t)_{t \in [0,T]}$ can be decomposed in a  pair of risk measures, acting before and after $\tau$ and we characterize its properties giving suitable assumptions on the driver of the BSDE. Furthermore, we prove an inequality satisfied by the penalty term associated to the robust representation of $(\rho_t)_{t \in [0,T]}$ and we {\color{blue}discuss the dynamic entropic risk measure case, providing examples where it is possible to write explicitly its decomposition and simulate it numerically.}
\end{abstract}

%

%
%

\maketitle


\smallskip

\noindent \textbf{AMS 2010:} 60H30, 91G80, 60J75, 60G44.

\thetavartheta

\smallskip

\section{Introduction} \label{sec:intro}
Risk measures have been introduced in an axiomatic way in {\color{blue} the seminal papers of \citet{artzner1999:coherentriskmeas} and \citet{delbaen2002}} in order to quantify the riskiness of financial positions.
\color{blue}
More precisely, the aforementioned authors introduced the so called coherent risk measures that can be interpreted as capital margins to be deposited to cover the riskiness of a position or, better, as the minimal cash amounts to be added to a position to make it acceptable whenever considered unacceptable.

Several extensions of coherent risk measures have been provided in the literature, both generalizing the given axioms of coherence and moving from a static to a dynamic framework.

Starting from the generalization at the level of the axioms, motivated by liquidity reasons \citet{follmer2002convex} and \citet{frittelli-rg2002} introduced convex risk measures by replacing subadditivity and positive homogeneity with the weaker axiom of convexity. Further extensions of coherent and convex risk measures can be also found in \citep{elkaroui2009:subaddriskmeas,frittelli-maggis2011} and in the references therein.

Differently from the first works on risk measures where a one-period model was considered, dynamic risk measures were introduced to evaluate the riskiness of financial positions at any time before the time horizon, hence taking into account a dynamic framework.
An early work discussing dynamic risk measures as a collection of conditional risk measures indexed by time (the concept adopted also in this paper) is the work by \citet{wang1999:dynrisk}, where only finite probability spaces are considered. Papers by \citet{riedel2004:dynamic, roorda2005:coherent} (considering finite probability spaces) and \citet{artzner2007:cohermultiper, delbaen2006structure} (with general probability spaces) followed, studying coherent dynamic risk measures in a discrete-time setting. Convex dynamic risk measures in a discrete-time setting were analyzed in \citet{cheridito2006:dynamic, detlefsen-scandolo2005, follmer2006:convexrisk}. Among the first works discussing dynamic risk measures in a continuous-time setting we can cite \citet{barrieu2009:pricing,bion-nadal2008,rosazza:riskgexpect}.
A key issue on dynamic risk measures consists then in the inter-temporal relations or, better, in the so called time-consistency of the risk measure and many of the aforementioned works deal with this topic.

Various approaches can be used to study dynamic risk measures: one of them is based on their strong connection with \emph{Backward Stochastic Differential Equations} (\emph{BSDEs}, for short). These stochastic equations have been extensively studied starting from the seminal paper of \citet{pardouxpeng1990:BSDE}\footnote{{\color{blue}An earlier study of linear BSDEs was presented in \citet{bismut1973:conjconv}.}} and nowadays are a well established and powerful tool in Mathematical Finance, as shown for instance in \citet{elkaroui1997:BSDEs}.

The connection between dynamic risk measures and BSDEs arises from the theory of \emph{conditional $g$-expectations} (or \emph{nonlinear expectations}), introduced by \citet{peng1997:BSDEsandgexp} by means of the solution of a BSDE driven by a Wiener process with generator $g$, and this link was studied in \citet{barrieu2009:pricing} and \citet{rosazza:riskgexpect}. In these papers the authors study also the relationship between the properties of the driver $g$ of the BSDE and the properties of the induced dynamic risk measure, as translation invariance, subadditivity and convexity, to name a few. Moreover, strong time consistency of the induced dynamic risk measure is established \emph{via} the flow property of BSDEs. Conditions on the driver $g$ have been considerably weakened in \citet{jiang2008:gexpect}, preserving the main properties of the associated dynamic risk measure.

Another important paper on the theory of nonlinear expectations by \citet{coquet2002:gexpect} provides the groundwork to answer the following question: under which assumptions is it possible to find a driver such that the associated BSDE induces a given dynamic risk measure? The answer has been proved to be affirmative in \citet{rosazza:riskgexpect} in a Brownian filtration framework. In this context the theory of $g$-expectations has been adopted also to represent, under suitable assumptions, the penalty term associated to a dynamic risk measure (see \citet{delbaen2010:reprpenalty}) and to study cash-subadditive risk measures in \citet{elkaroui2009:subaddriskmeas}.

Nonlinear expectations in the framework of BSDEs driven by a Brownian motion and an independent Poisson random measure were studied first in \citet{royer2006:nonlinexp} and the connection between dynamic risk measures and this class of BSDEs has been analyzed in \citet{quenez:BSDEJ}. In this framework, too, the theory of $g$-expectations has been adopted also to represent, under suitable assumptions, the penalty term associated to a dynamic risk measure (see \citet{tangwei2012:representation}, further generalized by \citet{laevenstadje2014:robustportfolio} to the case where the sources of noise are a Brownian motion and an independent real-valued marked point process).

The aim of the paper is to further extend results on dynamic risk measures induced by BSDEs, defining them and studying their properties, including time-consistency, in the framework of enlargement of filtration and under very general assumptions on the driver of the BSDE.
\normalcolor
More precisely, we show how to induce a dynamic risk measure from a BSDE whose driving noise is given by a Brownian motion and a {\color{blue} random counting measure. More specifically, to simplify the exposition, we assume that this random measure is a Dirac delta concentrated at a pair given by a random time and a random mark associated to it. However, it is possible to generalize all the arguments provided in this work to the case where the random counting measure is associated to a marked point process with a finite number of jumps in a fixed time interval.} In terms of the underlying information flow, this situation corresponds to a progressive enlargement of a Brownian reference filtration with information brought by the occurrence of random events at random times. This may describe, for instance, the presence of defaults in a financial market and the different behavior of risk measures before, between and after consecutive defaults.
{\color{blue}The theory of enlargement of filtration is deeply studied in the literature, starting from the seminal works of J.~Jacod, T.~Jeulin and M.~Yor (see, e.g., \citep{jacod1985:grossissement,jeulin1980:grossissement,jeulinyor1985:grossissements}).} We also refer the reader to~\citet{aksamitjeanblanc2017:filtrenl} (and the references therein) for an up-to-date account of enlargement of filtration theory and its applications to finance.

The aforementioned class of BSDEs with jumps that we are interested in was introduced by Kharroubi and Lim in \citep{kharroubi:progrenl}. There the authors have proved existence and uniqueness results, connecting such BSDEs with jumps to systems of Brownian BSDEs. This method enables a decomposition of the processes giving a solution to the former BSDEs into corresponding processes that are solution of the latter ones between each pair of consecutive jump times.

Here we show that dynamic risk measures induced by these BSDEs with jumps admit a similar decomposition into two risk measures, one before and the other after the default time. Furthermore, we prove that standard properties of dynamic risk measures are guaranteed by similar properties of the driver of these BSDEs and, finally, that induced dynamic risk measures are time-consistent. From a financial point of view, the decomposition of the \emph{global} risk measure into different \emph{local} ones seems to be reasonable. Before and after a default time, indeed, the risk measure should be updated in order to take into account the new information. A priori, there is no reason to impose that the risk measure remains unchanged.

The paper is organized as follows. In Section \ref{sec: decomposition-risk-measures} we introduce the reference and the enlarged filtrations we will deal with, {\color{blue} we provide the main assumptions of this paper and we give a brief recapitulation of the main results on the class of BSDEs with jumps introduced in~\citet{kharroubi:progrenl}. We also prove a new result on \emph{a priori} estimates on these BSDEs.
Then, we introduce dynamic risk measures induced by this class of BSDEs with jumps on a progressively enlarged filtration and we provide the decomposition of these dynamic risk measures.} The properties of these risk measures, {\color{blue} including time-consistency}, are studied in Section \ref{sec: properties-rm}, while the impact of enlargement of filtration on their dual representation is investigated in Section \ref{sec: representation-rm}. {\color{blue} In Section~\ref{sec:example} we discuss in detail the dynamic entropic risk measure defined on the progressively enlarged filtration: we show how to construct explicitly this type of risk measure, pointing out that, under our framework, it is possible to update it, upon the arrival of the information provided by the random time and the random mark, by allowing the associated risk-tolerance parameter -- or, equivalently, the risk-aversion parameter -- to change. We also provide closed formulas and numerical simulations for the risk evaluation of two specific financial claims. In particular, as will be explained in Section~\ref{sec:example}, the first claim has a structure that can be considered only in the framework of progressive enlargement of filtration and this clarifies, once more, the features of the dynamic risk measures studied in this work that, to the best of our knowledge, have not yet been considered in the literature.}

\subsection{Notation}
We collect here for the readers' convenience all the relevant notations and conventions that will be used throughout the paper.

The symbol $\R^+$ is a shorthand to denote the set $(0,+\infty)$. We denote by $\int_a^b$ the integral on the interval $(a,b]$, for $a, b \geq 0$.

Given a non-empty topological space $\mathcal{X}$, $\cB(\mathcal{X})$ denotes the Borel $\sigma$-algebra on $\mathcal{X}$ and we define $\dB(\mathcal{X}) \coloneqq \{f: \mathcal{X} \to \R, \text{Borel-measurable}\}$; we equip this space with the pointwise convergence topology.
We denote by $\cM_1(\cA)$ the collection of probability measures on a measurable space $(\Omega, \cA)$.

Given a filtration $\bbF$ on some probability space, $\cP(\bbF)$ (resp. $\cO(\bbF)$, $\prog(\bbF)$) denotes the $\bbF$-predictable (resp. $\bbF$-optional, $\bbF$-progressive) $\sigma$-algebra on $\Omega \times [0,+\infty)$, i.e., the $\sigma$-algebra generated by left-continuous (resp. right-continuous, progressively measurable) $\bbF$-adapted processes. The class of $\bbF$-predictable (resp. $\bbF$-optional, $\bbF$-progressive) processes will be indicated by $\cP_\bbF$ (resp. $\cO_\bbF$, $\prog_\bbF$). Finally, the symbol $\cP_\bbF(\R^+,E)$ (resp. $\cO_\bbF(\R^+,E)$, $\prog_\bbF(\R^+,E)$) denotes the class of indexed processes $Y$ such that the map {\color{blue}$\Omega \times [0,+\infty) \times \R^+ \times E \ni (\omega, t, \theta, e) \mapsto Y_t(\omega, \theta, e)$} is $\cP(\bbF) \otimes \cB(\R^+) \otimes \cB(E)$-measurable (resp. $\cO(\bbF) \otimes \cB(\R^+) \otimes \cB(E)$-measurable, $\prog(\bbF) \otimes \cB(\R^+) \otimes \cB(E)$-measurable).

\section{Enlargement of filtration, BSDEs and related risk measures} \label{sec: decomposition-risk-measures}
\subsection{Setting of the problem}
Let $T > 0$ be a fixed time horizon and $(\Omega, \cA, \bbP)$ be a complete probability space, supporting a standard $d$--dimensional Brownian motion $W = (W_t)_{t \geq 0}$, a random variable $\tau$ with values in $\R^+$, describing a \emph{default time}, and an associated \emph{mark}, i.e. a random variable $\zeta$ with values in a measurable subset $E$ of $\R^m$. We denote by $\bbF = (\cF_t)_{t \geq 0}$ the completed natural filtration generated by $W$. To be more precise, we set $\cF_t \coloneqq \sigma(W_s \colon 0 \leq s \leq t) \lor \cN$ for all $t \geq 0$, where $\cN \subset \cA$ is the collection of all $\bbP$--null sets. $\bbF$ plays the role of a \emph{reference filtration}, that we are going to enlarge with the information coming from the default time and the associated random mark.

We will work throughout this paper under the following fundamental assumption.
\begin{assumption}[Density hypothesis]\label{hp:density}
	For any $t \geq 0$, the conditional distribution of the pair $(\tau, \zeta)$ given $\cF_t$ is {\color{blue} equivalent to} the Lebesgue measure on $\R^+ \times E$. In particular, there exists a strictly positive $\gamma \in \cP_\bbF(\R^+, E)$ such that
	\begin{equation*}
			\bbP\bigl((\tau, \zeta) \in C \mid \cF_t \bigr) = \int_C \gamma_t(\theta, e) \, \dd \theta \, \dd e, \quad C \in \cB(\R^+) \otimes \cB(E), \, t \geq 0.
	\end{equation*}
\end{assumption}

\begin{remark}
The importance of the density hypothesis will be underlined extensively in what follows. To start, we want to emphasize that, since we are assuming that $\gamma > 0$, our assumption is stronger than Jacod's one~\citep{jacod1985:grossissement} (also adopted in \citep{kharroubi:progrenl}), i.e., the existence of a deterministic measure on $\R^+ \times E$ that dominates all the conditional laws. It is also slightly more stringent than the $(\cE)$-hypothesis in~\citep{callegaro2013:carthaginian}, i.e., the equivalence of all the conditional laws to the law of $(\tau, \zeta)$.

In fact, our assumption implies the $(\cE)$-hypothesis, since, applying the Fubini-Tonelli theorem and recalling that $\bigl(\gamma_t(\theta,e)\bigr)$ is {\color{blue} an $\bbF$-martingale} for any $(\theta,e) \in \R^+\times E$ (see \citep[Lemme (1.8)]{jacod1985:grossissement}), we have
\begin{equation*}
	\bbP\bigl((\tau, \zeta) \in C\bigr) = \bbE\bigl[\bbE[\ind_C(\tau, \zeta) \mid \cF_T]\bigr] = \bbE\biggl[\int_C \gamma_T(\theta, e) \, \dd \theta \, \dd e\biggr] = \int_C \gamma_0(\theta,e) \, \dd \theta \, \dd e.
\end{equation*}
In particular, the law of $(\tau, \zeta)$ is equivalent to the Lebesgue measure on $\R^+ \times E$.
\end{remark}

Let us define for all $t \geq 0$ the following families of $\sigma$-algebras:
\begin{itemize}
	\item $\cD_{t} \coloneqq \sigma(\ind_{\tau \leq s}, \, 0 \leq s \leq t) \lor \sigma(\zeta \ind_{\tau \leq s}, \, 0 \leq s \leq t)$;
	\item $\cG_t \coloneqq \bigcap_{s > t} \bigl(\cF_s \lor \cD_s\bigr)$.
	\item $\cH_t \coloneqq \cF_t \lor \sigma(\tau,\zeta)$.
\end{itemize}
We call $\bbG \coloneqq (\cG_t)_{t \geq 0}$ the \emph{progressively enlarged filtration} and $\bbH \coloneqq (\cH_t)_{t \geq 0}$ the \emph{initially enlarged filtration}.

\begin{remark}\label{rem:filtrations}
It is known, (see, e.g., \citep{callegaro2013:carthaginian,jeanblanc2009:progrenl}), that $\bbF \subset \bbG \subset \bbH$ and that the filtration $\bbG$ coincides with $\bbF$ on the event $\{t < \tau\}$ and with $\bbH$ on the event $\{t \geq \tau\}$. This fact will be fundamental in the sequel.
\end{remark}

It is possible to give a decomposition of any $\cG_t$-measurable random variable and of $\bbG$-predictable processes into indexed $\cF_t$-measurable random variables and $\bbF$-predictable  processes, where the indexes are given by the jump time and the associated random mark. These results are given in \citep{callegaro2013:carthaginian,pham:progrenl} (or are easy generalizations of them) and we recall them here for the readers' convenience.

\begin{lemma}[{\citep[Prop. 2.8]{callegaro2013:carthaginian} and \citep[Lemma 2.1]{pham:progrenl}}]\label{lemma:procdecomp}
	It holds that:
	\begin{enumerate}
	\item For any $t \geq 0$, a random variable $\xi$ is $\cG_t$-measurable if and only if it is of the form
	\begin{equation*}
		\xi(\omega) = \xi^0(\omega) \ind_{t < \tau(\omega)} + \xi^1(\omega, \tau(\omega),\zeta(\omega)) \ind_{t \geq \tau(\omega)},
	\end{equation*}
	for some $\cF_t$-measurable random variable $\xi^0$ and a family of $\cF_t \otimes \cB(\R^+) \otimes \cB(E)$-measurable random variables $\xi^1(\cdot, \theta, e), \, \theta \in [0,t], \, e \in E$.
	\item A process $Y=(Y_t)_{t \geq 0}$ is $\bbG$-predictable if and only if it is of the form
	\begin{equation*}
		Y_t = Y_t^0 \ind_{t \leq \tau} + Y_t^1(\tau, \zeta) \ind_{t > \tau}, \quad t \geq 0,
	\end{equation*}
	where $Y^0 \in \cP_\bbF$ and $Y^1 \in \cP_\bbF(\R^+, E)$.
	\end{enumerate}	
\end{lemma}

Under the density assumption it is also possible to give a similar decomposition for $\bbG$-optional processes.
\begin{lemma}[{\citep[Lemma 2.1]{pham:progrenl} and \citep[Th. 7.5]{song2014:optionalsplitting}}]\label{lemma:optprocdecomp}
	Any $\bbG$-optional process $Y=(Y_t)_{t \geq 0}$ can be decomposed as
	\begin{equation}\label{eq:optsplit}
		Y_t = Y_t^0 \ind_{t < \tau} + Y_t^1(\tau, \zeta) \ind_{t \geq \tau}, \quad t \geq 0,
	\end{equation}
	where $Y^0 \in \cO_\bbF$ and $Y^1 \in \cO_\bbF\bigl(\R^+,E\bigr)$.
\end{lemma}

\color{blue}
\begin{remark}
Equation~\eqref{eq:optsplit}, called \emph{optional splitting formula} in \citep{song2014:optionalsplitting}, is in general not valid, as a famous example by Barlow shows (see \citep{barlow1978:honest, song2014:optionalsplitting}). For a thorough study on the subject, we refer the reader to \citep{song2014:optionalsplitting}, where various conditions for~\eqref{eq:optsplit} to hold are analyzed.
\end{remark}
\normalcolor

\subsection{BSDEs with jumps}\label{subs:BSDEJ}
We introduce now BSDEs driven by the Wiener process $W$ and {\color{blue}the} following random measure $\mu$ on $\R^+ \times E$:
\begin{equation*}
	\mu\bigl(\omega; \, (0,t] \times B\bigr) \coloneqq \ind_{\tau(\omega) \leq t} \, \ind_{\zeta(\omega) \in B}, \quad t > 0, \, B \in \cB(E), \, \omega \in \Omega.
\end{equation*}
Such a stochastic differential equation will be called BSDEJ (short for BSDE with jumps). All the results contained in this Subsection, with the exception of Proposition~\ref{prop:BSDEJcont}, are proved in~\citep{kharroubi:progrenl}, hence we will state them without proof for the readers' convenience. We also point out that all of them are valid under the weaker version of the density assumption adopted in~\citep{kharroubi:progrenl}, i.e., assuming that the $\cF_t$-conditional laws of $(\tau, \zeta)$ are all absolutely continuous with respect to the Lebesgue measure on $\R^+ \times E$.

We need, first, to introduce the spaces in which we will look for solutions of the BSDEJ.
\begin{itemize}
	\item $\cS_\bbG^\infty[a,b]$ (resp. $\cS_\bbF^\infty[a,b]$) is the set of real-valued processes $Y \in \prog_\bbG$ (resp. $Y \in \prog_\bbF$), such that:
	\begin{equation*}
		\norm{Y}_{\cS^\infty[a,b]} \coloneqq \esssup_{\omega \in \Omega, \, t \in [a,b]} |Y_t(\omega)| < \infty.
	\end{equation*}
	\item $\dL^2_\bbG[a,b]$ (resp. $\dL^2_\bbF[a,b]$) is the set of $\R^d$-valued processes $Z \in \cP_\bbG$ (resp. $Z \in \cP_\bbF$) such that
	\begin{equation*}
		\norm{Z}_{\dL^2[a,b]} \coloneqq \biggl(\bbE\biggl[\int_a^b |Z_t|^2 \, \dd t\biggr]\biggr)^{\frac 12} < \infty.
	\end{equation*}
	\item $\dL^2(\mu)$ is the set of real-valued $E$-indexed processes $U(\cdot) \in \cP_\bbG$ such that
	\begin{equation*}
		\norm{U}_{\dL^2(\mu)} \coloneqq \biggl(\bbE\biggl[\int_0^T \int_E |U_s(e)|^2 \, \mu(\dd s \, \dd e)\biggr]\biggr)^{\frac 12} < \infty.
	\end{equation*}
\end{itemize}

Solutions to our BSDEJ will be triples $(Y,Z,U) \in S^\infty_\bbG[0,T] \times \dL^2_\bbG[0,T] \times \dL^2(\mu)$ satisfying
\begin{equation}\label{eq:BSDEJ}
	Y_t = \xi + \int_t^T g(s, Y_s, Z_s, U_s(\cdot)) \, \dd s - \int_t^T Z_s \, \dd W_s - \int_t^T \int_E U_s(e) \, \mu(\dd s \, \dd e), \quad t \in [0,T],
\end{equation}
where:
\begin{itemize}
	\item $\xi$ is a $\cG_T$-measurable random variable.
	\item $g \colon \Omega \times [0,T] \times \R \times \R^d \times \dB(E) \to \R$ is a $\cP(\bbG) \otimes \cB(\R) \otimes \cB(\R^d) \otimes \cB\bigl(\dB(E)\bigr)$-measurable map.
\end{itemize}
Notice that from Lemma~\ref{lemma:procdecomp} we get the following decompositions for $\xi$ and $g$:
\begin{gather}
	\xi = \xi^0 \ind_{T < \tau} + \xi^1(\tau,\zeta) \ind_{T \geq \tau}, \label{eq:xidecomp} \\
	g(t,y,z,u) = g^0(t,y,z,u) \ind_{t \leq \tau} + g^1(t,y,z,u,\tau,\zeta) \ind_{t > \tau}. \label{eq:gdecomp}
\end{gather}

\begin{theorem}[Existence theorem, {\citep[Th. 3.1]{kharroubi:progrenl}}]\label{th:BSDEJexist}
	Assume that for all $(\theta, e) \in \R^+ \times E$ the Brownian BSDE
	\begin{multline}\label{eq:BSDE1}
		Y_t^1(\theta,e) = \xi^1(\theta,e) + \int_t^T g^1(s,Y_s^1(\theta,e), Z_s^1(\theta,e), 0, \theta, e) \, \dd s \\
		 - \int_t^T Z_s^1(\theta,e) \, \dd W_s, \quad \theta \land T \leq t \leq T,
	\end{multline}
	admits a solution $\bigl(Y^1(\theta,e), Z^1(\theta,e)\bigr) \in \cS_\bbF^\infty[\theta \land T, T] \times \dL^2_\bbF[\theta \land T, T]$, and that the Brownian BSDE
	\begin{equation}\label{eq:BSDE0}
		Y_t^0 = \xi^0 + \int_t^T g^0(s, Y_s^0, Z_s^0,Y_s^1(s, \cdot) - Y_s^0) \, \dd s - \int_t^T Z_s^0 \, \dd W_s, \quad 0 \leq t \leq T,
	\end{equation}
	admits a solution $(Y^0, Z^0) \in \cS_\bbF^\infty[0, T] \times \dL^2_\bbF[0, T]$. Assume, moreover, that the pair $(Y^1, Z^1)$ is such that $Y^1 \in \prog_\bbF\bigl(\R^+ \times E \bigr)$ and $Z^1 \in \cP_\bbF\bigl(\R^+ \times E \bigr)$.
	
	If these solutions satisfy
	\begin{gather}
		\norm{Y^0}_{\cS^\infty[0,T]} < +\infty, \\
		\sup_{(\theta,e) \in \R^+ \times E} \norm{Y^1(\theta,e)}_{\cS^\infty[\theta \land T ,T]} < +\infty \\
		\bbE\biggl[\int_{\R^+ \times E} \biggl(\int_0^{\theta \land T} |Z_s^0|^2 \, \dd s + \int_{\theta \land T}^T |Z_s^1(\theta,e)|^2 \, \dd s\biggr) \, \gamma_T(\theta,e) \, \dd \theta \, \dd e\biggr] < +\infty,
	\end{gather}
	
	then the BSDEJ~\eqref{eq:BSDEJ} admits a solution $(Y,Z,U) \in \cS_\bbG^\infty[0,T] \times \dL^2_\bbG[0,T] \times \dL^2(\mu)$ given by
	\begin{equation}\label{eq:BSDEJsol}
		\left\{
		\begin{aligned}
			&Y_t = Y_t^0 \, \ind_{t < \tau} + Y_t^1(\tau,\zeta) \, \ind_{t \geq \tau}, \\
			&Z_t = Z_t^0 \, \ind_{t < \tau} + Z_t^1(\tau,\zeta) \, \ind_{t \geq \tau}, \\
			&U_t(\cdot) = U_t^0(\cdot) \, \ind_{t < \tau} = \Bigl[Y_t^1(t, \cdot) - Y_t^0\Bigr] \, \ind_{t < \tau}.
		\end{aligned}
		\right.
	\end{equation}
\end{theorem}

\color{blue}
This result heavily relies on the assumption that each Brownian BSDE admits a solution. In the literature several cases are studied to ensure this fact, among which we are going to present the \emph{quadratic case}. This situation was analyzed first in~\citep{kobylanski2000:BSDEs} and applied to the BSDEJ discussed here in~\citep{kharroubi:progrenl}.
\begin{proposition}[{\citep[Prop. 3.1]{kharroubi:progrenl}}]\label{prop:BSDEJquadexist}
Suppose that, together with the density assumption, we have that:
\begin{enumerate}
\item The \emph{compensator} $\lambda_t(e) \, \dd t \, \dd e$ of the random counting measure $\mu$ is such that the process $\Bigl(\int_E \lambda_t(e) \, \dd e\Bigr)_{t \geq 0}$ is bounded.
\item The terminal condition $\xi$ is bounded, i.e., there exists $C > 0$ such that $|\xi| \leq C$, $\bbP$-a.s.
\item The generator $g$ is quadratic in $z$, i.e., there exists $K > 0$ such that, for all $(t,y,z,u) \in [0,T] \times \R \times \R^d \times \dB(E)$,
\begin{equation*}
|g(t,y,z,u)| \leq K \biggl(1 + |y| + \norm{z}^2 + \int_E |u(e)| \lambda_t(e) \, \dd e \biggr).
\end{equation*}
\item For any $R>0$ there exists a function $\omega_R$ such that $\displaystyle \lim_{\epsilon \to 0} \omega_R(\epsilon) = 0$ and
\begin{equation*}
|g\bigl(t,y,z,(u(e) - y)_{e \in E}\bigr) - g\bigl(t,y',z',(u(e) - y)_{e \in E}\bigr)| \leq \omega_R(\epsilon),
\end{equation*}
for all $t \in [0,T]$, $y,y' \in \R$, $z,z'\in \R^d$, $u \in \dB(E)$ such that $|y|, |y'|, \norm{z}, \norm{z'} \leq R$ and $|y-y'| + \norm{z-z'} \leq \epsilon$.
\end{enumerate}
Then the BSDEJ~\eqref{eq:BSDEJ} admits a solution.
\end{proposition}
\normalcolor

Before stating the comparison theorem, let us introduce what follows.
\begin{itemize}
\item $\underline\xi, \, \overline\xi \in \dL^\infty(\cG_T)$.
\item $\underline g, \, \overline g \colon \Omega \times [0,T] \times \R \times \R^d \times \dB(E) \to \R$, two $\cP(\bbG) \otimes \cB(\R) \otimes \cB(\R^d) \otimes \cB\bigl(\dB(E)\bigr)$-measurable maps.
\item $(\underline g^0, \underline g^1)$ and $(\overline g^0, \overline g^1)$, the decompositions of $\underline g$ and $\overline g$, as appearing in~\eqref{eq:gdecomp}.
\item $(\underline Y, \underline Z, \underline U)$ and $(\overline Y, \overline Z, \overline U)$ the solutions of the BSDEJ~\eqref{eq:BSDEJ} with pair driver/terminal condition given by $(\underline g, \underline \xi)$ and $(\overline g, \overline \xi)$ respectively.
\item $(\underline Y^0, \underline Y^1)$ (resp. $(\underline Z^0, \underline Z^1)$, $(\underline U^0, \underline U^1)$, $(\overline Y^0, \overline Y^1)$, $(\overline Z^0, \overline Z^1)$, $(\overline U^0, \overline U^1)$), the decompositions of $\underline Y$ (resp. $\underline Z$, $\underline U$, $\overline Y$, $\overline Z$, $\overline U$).
\item $\underline G^0(t,y,z) \coloneqq \underline g^0(t,y,z,\underline Y^1_t(t, \cdot) - y)$, $\overline G^0(t,y,z) \coloneqq \overline g^0(t,y,z,\overline Y^1_t(t, \cdot) - y)$.
\item $\underline G^1(t,y,z) \coloneqq \underline g^1(t,y,z,0)$, $\overline G^1(t,y,z) \coloneqq \overline g^1(t,y,z,0)$.
\end{itemize}

We need the following definitions:
\begin{definition}[{\citep[Def. 4.1]{kharroubi:progrenl}}]
We say that a driver $f \colon \Omega \times [0,T] \times \R \times \R^d \to \R$ satisfies a \emph{comparison theorem for Brownian BSDEs} if for any bounded $\bbG$-stopping times $\nu_2 \geq \nu_1$, any driver $f' \colon \Omega \times [0,T] \times \R \times \R^d \to \R$ and any $\cG_{\nu_2}$-measurable random variables $X, \, X'$ such that $f \leq f'$ and $X \leq X'$ (resp. $f \geq f'$ and $X \geq X'$), we have $Y \leq Y'$ (resp. $Y \geq Y'$) on $[\nu_1, \nu_2]$, where $(Y,Z)$, $(Y',Z')$ are solutions in $\cS_\bbG^\infty[0,T] \times \dL^2_\bbG[0,T]$ to the BSDEs:
\begin{gather*}
	Y_t = X + \int_t^{\nu_2} f(s, Y_s, Z_s) \, \dd s - \int_t^{\nu_2} Z_s \, \dd W_s, \quad t \in [\nu_1, \nu_2],
	\\
	Y'_t = X' + \int_t^{\nu_2} f'(s, Y'_s, Z'_s) \, \dd s - \int_t^{\nu_2} Z'_s \, \dd W_s, \quad t \in [\nu_1, \nu_2].
\end{gather*}
\end{definition}


To state the comparison theorem and all the remaining results of this paper, we need the so called \emph{immersion hypothesis}, also known as Hypothesis (H), to be in force.

\begin{assumption}[Immersion hypothesis]\label{hp:martingale}
	 The filtration $\bbF$ is immersed in the filtration $\bbG$ under $\bbP${\color{blue}, i.e., any $(\bbP, \bbF)$-martingale is a $(\bbP, \bbG)$-martingale.}
\end{assumption}

A characterization of the immersion property is given, e.g., in \citep[Proposition 3.6 (b)]{aksamitjeanblanc2017:filtrenl}.

\begin{theorem}[Comparison theorem, {\citep[Th. 4.1]{kharroubi:progrenl}}]\label{th:BSDEJcomp}
Suppose that $\underline\xi \leq \overline\xi, \, \bbP$-a.s. \, Suppose, moreover, that for all $(t,y,z) \in [0,T] \times \R \times \R^d$
\begin{gather*}
	\underline G^0(t,y,z) \leq \overline G^0(t,y,z), \quad \bbP\text{-a.s.},
	\\
	\underline G^1(t,y,z) \leq \overline G^1(t,y,z), \quad \bbP\text{-a.s.},
\end{gather*}
and that the drivers $\underline G^0, \, \underline G^1$ or $\overline G^0, \, \overline G^1$ satisfy a comparison theorem for Brownian BSDEs.

If, in addition, $\underline U_t = \overline U_t = 0$ on $\{t > \tau \}$, then we have that
\begin{equation*}
	\underline Y_t \leq \overline Y_t, \quad t \in [0,T], \, \bbP\text{-a.s.}
\end{equation*}
\end{theorem}

\color{blue}
Also this result, as the existence one, is not directly usable, since the condition on the drivers of the Brownian BSDEs of type~\eqref{eq:BSDE0} depends on the solutions of the Brownian BSDEs of type~\eqref{eq:BSDE1}. It is possible, however, to provide an example where uniqueness of the solution to the BSDEJ~\eqref{eq:BSDEJ} is shown through Theorem~\ref{th:BSDEJcomp}. It is developed in~\citep{kharroubi:progrenl} and based on a comparison theorem for quadratic BSDEs with convex generators, given in~\citep{briandhu2008:quadbsde}.

\begin{theorem}[{\citep[Th. 4.2]{kharroubi:progrenl}}]\label{th:BSDEJquaduniq}
Suppose that, together with the density assumption and the immersion hypothesis, we have that:
\begin{enumerate}
\item The \emph{compensator} $\lambda_t(e) \, \dd t \, \dd e$ of the random counting measure $\mu$ is such that the process $\Bigl(\int_E \lambda_t(e) \, \dd e\Bigr)_{t \geq 0}$ is bounded.
\item For all $(t,y,u) \in [0,T] \times \R \times \dB(E)$, the map $z \mapsto g(t,y,z,u)$ is convex.
\item The generator $g$ is Lipschitz with respect to $y$, i.e., there exists $L > 0$ such that
\begin{equation*}
|g\bigl(t,y,z,(u(e) - y)_{e \in E}\bigr) - g\bigl(t,y',z,(u(e) - y')_{e \in E}\bigr)| \leq L |y-y'|,
\end{equation*}
for all $t \in [0,T]$, $y,y' \in \R$, $z \in \R^d$, $u \in \dB(E)$.
\item The generator $g$ is quadratic in $z$, i.e., there exists $K > 0$ such that, for all $(t,y,z,u) \in [0,T] \times \R \times \R^d \times \dB(E)$,
\begin{equation*}
|g(t,y,z,u)| \leq K \biggl(1 + |y| + \norm{z}^2 + \int_E |u(e)| \lambda_t(e) \, \dd e \biggr).
\end{equation*}
\item $g(t, \cdot, \cdot, u) = g(t, \cdot, \cdot, 0)$ for all $u \in \dB(E)$ and all $t \in (\tau_n \land T, T]$.
\end{enumerate}
Then the BSDEJ~\eqref{eq:BSDEJ} admits at most one solution.
\end{theorem}
\normalcolor

We conclude this Subsection giving a result connecting \emph{a priori} estimates for the first components of solutions to Brownian BSDEs~\eqref{eq:BSDE1} and~\eqref{eq:BSDE0} with an estimate for the first component of the solution to the BSDEJ~\eqref{eq:BSDEJ}. This result was not given in~\citep{kharroubi:progrenl} and will be important to get continuity of the dynamic risk measure we aim to study (see Section~\ref{sec: properties-rm}).

\begin{proposition}\label{prop:BSDEJcont}
Let the assumptions of Theorems~\ref{th:BSDEJexist} and~\ref{th:BSDEJcomp} hold. Let $\bar \xi, \, \hat \xi \in \dL^\infty(\cG_T)$, with associated decompositions:
\begin{gather*}
\bar \xi = \bar \xi^0 \, \ind_{T < \tau} + \bar \xi^1(\tau,\zeta) \, \ind_{T \geq \tau}, \\
\hat \xi = \hat \xi^0 \, \ind_{T < \tau} + \hat \xi^1(\tau,\zeta) \, \ind_{T \geq \tau}.
\end{gather*}
Denote by $(\bar Y, \bar Z, \bar U)$ (resp. $(\hat Y, \hat Z, \hat U)$) the solution to the BSDEJ with driver $g$ and terminal condition $\bar \xi$ (resp. $\hat \xi$), defined by the pairs $(\bar Y^0, \bar Z^0)$ and $(\bar Y^1, \bar Z^1)$ (resp. $(\hat Y^0, \hat Z^0)$ and $(\hat Y^1, \hat Z^1)$) as in~\eqref{eq:BSDEJsol}.

Suppose, moreover, that for each $(\theta,e) \in [0,T] \times E$:
\begin{equation}\label{eq:BSDE01apriori}
\begin{gathered}
	\norm{\bar Y^0 - \hat Y^0}_{\cS^\infty[0,T]} \leq K^0 \norm{\bar \xi^0 - \hat \xi^0}_{\dL^\infty}, \\	
	\norm{\bar Y^1(\theta,e) - \hat Y^1(\theta,e)}_{\cS^\infty[\theta,T]} \leq K^1(\theta,e) \norm{\bar \xi^1(\theta,e) - \hat \xi^1(\theta,e)}_{\dL^\infty},
\end{gathered}
\end{equation}
for some constants $K^0, \, K^1(\theta,e) > 0$.

If we have that $M \coloneqq \max\{M_0, M_1\} < +\infty$, where
\begin{gather*}
M_0 \coloneqq K^0\norm{\bar \xi^0 - \hat \xi^0}_{\dL^\infty}, \\
M_1 \coloneqq \sup_{(\theta,e) \in [0,T] \times E} K^1(\theta,e) \norm{\bar \xi^1(\theta,e) - \hat \xi^1(\theta,e)}_{\dL^\infty},
\end{gather*}
then the following estimate holds:
\begin{equation}\label{eq:deltaYestimate}
\norm{\bar Y - \hat Y}_{\cS^\infty[0,T]} \leq 2M.
\end{equation}
\end{proposition}

\begin{proof}
Fix $\bar \xi, \, \hat \xi \in \dL^\infty(\cG_T)$. Let us define the set
\begin{multline*}
A \coloneqq \{(\omega, \theta, e) \in \Omega \times [0,T] \times E \colon
\\
\max\{\sup_{t \in [0,T]} |\bar Y^0_t - \hat Y^0_t|, \sup_{t \in [\theta,T]} |\bar Y^1_t(\theta,e) - \hat Y^1_t(\theta,e)| \} \leq M \}.
\end{multline*}
Notice that $A \in \cF_T \otimes \cB(\R^+) \otimes \cB(E)$. Let us denote by $\tilde \Omega \coloneqq \{\omega \in \Omega \colon \bigl(\omega, \tau(\omega), \zeta(\omega)\bigr) \in A\} \in \cA$. We have that $\bbP(\tilde \Omega) = 1$, since, thanks to Assumption~\ref{hp:density}
\begin{align*}
\bbP(\tilde \Omega^c)
&= \bbE\bigl[\ind_{A^c}\bigl(\omega, \tau(\omega), \zeta(\omega)\bigr)\bigr] = \bbE\bigl[\bbE\bigl[\ind_{A^c}\bigl(\omega, \tau(\omega), \zeta(\omega)\bigr)\bigr] \mid \cF_T \bigr]
\\
&= \int_{[0,T] \times E} \bbE\bigl[\ind_{A^c}\bigl(\omega, \theta, e\bigr) \gamma_T(\theta,e)\bigr] \, \dd \theta \, \dd e
\end{align*}
and, given the definition of $A$ and~\eqref{eq:BSDE01apriori}, we have that $\ind_{A^c}\bigl(\cdot, \theta, e\bigr) \gamma_T(\theta,e) = 0, \, \bbP$-a.s., for all $(\theta,e) \in [0,T] \times E$, hence $\bbP(\tilde \Omega^c) = 0$.

Thanks to the decomposition of the processes $\bar Y$ and $\hat Y$ we have that for all $t \in [0,T]$
\begin{equation*}
|\bar Y_t - \hat Y_t| \leq |\bar Y^0_t - \hat Y^1_t| \ind_{t < \tau} + |\bar Y^1_t(\tau,\zeta) - \hat Y^1_t(\tau, \zeta)| \ind_{t \geq \tau}, \quad \bbP\text{-a.s.}
\end{equation*}
Since $\tilde \Omega$ is a set of full $\bbP$-measure, we also get:
\begin{equation*}
|\bar Y^0_t - \hat Y^0_t| \ind_{t < \tau} \leq M,
\quad
|\bar Y^1_t(\tau,\zeta) - \hat Y^1_t(\tau, \zeta)| \ind_{t \geq \tau} \leq M,
\quad
\bbP\text{-a.s.},
\end{equation*}
whence $|\bar Y_t - \hat Y_t| \leq 2M$, $\bbP$-a.s., for all $t \in [0,T]$, and the claim follows immediately.
\end{proof}

\subsection{Risk measure induced by the BSDEJ}
The aim of this Subsection is to introduce a \emph{dynamic risk measure} on the filtration $\bbG$, defined through the BSDEJ~\eqref{eq:BSDEJ}, and to decompose it in the same spirit of Lemma~\ref{lemma:procdecomp}. In this Subsection, we will work under the assumptions of Theorems~\ref{th:BSDEJexist} and~\ref{th:BSDEJcomp}, to have existence and uniqueness of the solution of the BSDEJ~\eqref{eq:BSDEJ}.

In our paper, a \emph{$\bbG$-dynamic risk measure} is a family of \emph{$\bbG$-conditional risk measures} $\rho_t \colon \dL^\infty(\cG_T) \to \dL^\infty(\cG_t)$, with $t \in [0,T]$. Following~\citep{rosazza:riskgexpect} and~\citep{detlefsen-scandolo2005}, we call a $\bbG$-conditional risk measure any map $\rho_t$ such that:
\begin{enumerate}
\item $\rho_t \colon \dL^\infty(\cG_T) \to \dL^\infty(\cG_t)$, for all $t \in [0,T]$;
\item $\rho_0$ is a \emph{static risk measure}, i.e., a functional $\rho_0 \colon \dL^\infty(\cG_T) \to \R$;
\item $\rho_T(\xi) = -\xi$, for all $\xi \in \dL^\infty(\cG_T)$.
\end{enumerate}
Clearly, we have similar definitions for $\bbF$- and $\bbH$-dynamic and conditional risk measures.

It is well known (see, e.g., \citep{barrieu2009:pricing, elkaroui2009:subaddriskmeas, quenez:BSDEJ, rosazza:riskgexpect}) that an important class of dynamic risk measures is the one induced by BSDEs. With this expression we mean that the following definition, as it is immediate to verify, gives birth to a family of conditional risk measures $\rho_t$ as defined above:
\begin{equation}\label{eq:dynriskmeas}
	\rho_t(\xi) \coloneqq Y_t^{-\xi}, \quad t \in [0,T], \, \xi \in \dL^\infty(\cG_T),
\end{equation}
where $Y^{-\xi} \coloneqq (Y_t^{-\xi})_{t \geq 0}$ is the first component of the solution of the BSDEJ~\eqref{eq:BSDEJ} with terminal condition $-\xi$. {\color{blue} The driver $g$ of the BSDEJ is fixed throughout this Section, therefore, in order to avoid an unnecessary heavy notation, we omit the dependence of $\rho$ on it.}

It is rather easy to see that, thanks to~\eqref{eq:BSDEJsol}, we can provide the following decomposition of $\rho$ into two dynamic risk measures, one acting on the event $\{t < \tau\}$ and the other on $\{t \geq \tau\}$.

\begin{proposition}\label{prop:rhodecomp}
	Under the assumptions of Theorems~\ref{th:BSDEJexist} and \ref{th:BSDEJcomp}, there exist an $\bbF$-dynamic risk measure $\rho^0 \coloneqq (\rho^0_t)_{t \in [0,T]}$ and an $\bbH$-dynamic risk measure $\rho^1 \coloneqq (\rho^1_t)_{t \in [0, T]}$,
	such that the dynamic risk measure $\rho$ defined in~\eqref{eq:dynriskmeas} admits the following decomposition:
	\begin{equation}\label{eq:rhodecomp}
		\rho_t(\xi) = \rho^0_t(\xi^0) \, \ind_{t < \tau} + \rho^1_t\bigl(\xi^1(\tau,\zeta)\bigr) \, \ind_{t \geq \tau}, \quad t \in [0,T], \, \xi \in \dL^\infty(\cG_T).
	\end{equation}
\end{proposition}

\begin{proof}
	Let $t \in [0,T]$ and $-\xi \in \dL^\infty(\cG_T)$ be fixed, and let $-\xi^0$ and $-\xi^1$ as in~\eqref{eq:xidecomp}. Under the assumptions of Theorems~\ref{th:BSDEJexist} and \ref{th:BSDEJcomp}, there exists a unique solution to the BSDEJ~\eqref{eq:BSDEJ} given by the triple $(Y,Z,U)$ defined in~\eqref{eq:BSDEJsol}. In particular:
	\begin{equation*}
		Y_t = Y_t^0 \, \ind_{t < \tau} + Y_t^1(\tau,\zeta) \, \ind_{t \geq \tau}, \quad t \in [0,T].
	\end{equation*}
	Since $Y^0$ and $Y^1(\tau,\zeta)$ are solutions to their {\color{blue} respective Brownian BSDEs~\eqref{eq:BSDE0} and~\eqref{eq:BSDE1}} (with terminal conditions $-\xi^0$ and $-\xi^1(\tau,\zeta)$ respectively), we can define for any $t \in [0,T]$ the maps $\rho^0_t$ and $\rho^1_t$ and easily verify that they are, respectively, $\bbF$- and $\bbH$-conditional risk measures:
	\begin{align*}
	\rho^0_t(\xi^0) &\coloneqq Y_t^0,
	\\
	\rho^1_t(\xi^1(\tau,\zeta)) &\coloneqq \bar\rho_t(\xi^1(\tau,\zeta)) \ind_{t < \tau} + Y_t^1(\tau,\zeta) \ind_{t \geq \tau}.
	\end{align*}
	Here, $\bar \rho_t \colon \dL^\infty(\cH_T) \to \dL^\infty(\cH_t)$, $t \in [0,T]$, is a family of conditional risk measures that is irrelevant to specify.
	The decomposition provided in~\eqref{eq:rhodecomp} follows immediately.
\end{proof}

\begin{remark}\label{rem:gexpect}
The construction of the dynamic risk measures $\rho^0$ and $\rho^1$ provided in the proof of the previous Proposition is reminiscent of $g$-conditional expectations, introduced by Peng in~\citep{peng1997:BSDEsandgexp}. Here, the peculiar structure of the Brownian BSDEs~\eqref{eq:BSDE1} and~\eqref{eq:BSDE0} does not allow to identify these dynamic risk measures as true $g^0$- and $g^1$-conditional expectations respectively. For instance, consider that the process $Y^1(\tau, \zeta)$ solves the BSDE~\eqref{eq:BSDE1} (with $(\theta,e) = (\tau, \zeta)$) only on the random time interval $[\tau \land T, T]$, hence we do not obtain a $g^1$-conditional expectation defined for all $t \in [0,T]$. Also, think about the fact that the solution $(Y^0, Z^0)$ of the BSDE~\eqref{eq:BSDE0} depends on the $\dB(E)$-valued process $\bigl(Y^1_s(s, \cdot)\bigr)_{s \geq 0}$, i.e., it is intertwined with the family of indexed processes that provide solutions to each of the BSDEs~\eqref{eq:BSDE1}, as $(\theta,e)$ varies.
\end{remark}

\begin{remark}\label{rem:rhoconstruction}
The dynamic risk measures $\rho^0$ and $\rho^1$ obtained thanks to~\eqref{eq:rhodecomp} have a very intuitive financial interpretation. On the one hand, the $\bbF$-conditional risk measure $\rho^0_t$ quantifies the risk associated to some financial position at time $t \in [0,T]$ in a world where the only available information is given by $\bbF$, i.e., no default has occurred up to time $t$. On the other hand, the conditional risk measure $\rho^1_t$ expresses the risk in a context where the default has occurred prior to (or at) time $t$, hence the information available is given by $\bbH$. The risk measure is updated to take into account the new information. This fact can be clearly seen in the example provided in Section~\ref{sec:example}, where we exhibit a dynamic entropic risk measure possessing this feature.
\end{remark}

Before concluding this Section, we want to make a small digression to explain how this reasoning can be reversed, i.e., how it is possible to construct a $\bbG$-dynamic risk measure starting from an appropriate pair of dynamic risk measures $\rho^0$ and $\rho^1$. We provide two possible examples of this construction. Notice that we are not assuming in what follows that the dynamic risk measures involved are induced by BSDEs.

\begin{example}
Let us be given an $\bbF$-dynamic risk measure $\rho^0$ and an $\bbH$-dynamic risk measure $\rho^1$. Then, we obtain a $\bbG$-dynamic risk measure $\rho$ defining:
\begin{equation*}
	\rho_t(\xi) \coloneqq \rho^0_t(\xi^0) \ind_{t < \tau} + \rho^1_t(\xi^1(\tau,\zeta)) \ind_{t \geq \tau},	\quad t \in [0,T], \, \xi \in \dL^\infty(\cG_T).
\end{equation*}
The only delicate property to check in order to ensure that $\rho$ is a true $\bbG$-dynamic risk measure is that $\rho_t(\xi)$ is a $\cG_t$-measurable random variable for any $t \in [0,T]$ and any $\xi \in \dL^\infty(\cG_T)$ (it is immediate to verify that it is essentially bounded).

Notice that, since $\rho^1$ is an $\bbH$-dynamic risk measure by assumption, $\rho_t^1(\xi^1(\tau,\zeta))$ is an $\cH_t$-measurable random variable. Hence, by (i) of~\citep[Prop. 2.7]{callegaro2013:carthaginian}, there exists an $\cF_t \otimes \cB(\R^+) \otimes \cB(E)$-measurable function $x_t \colon \Omega \times \R^+ \times E \to \R$ such that $\rho_t^1(\xi^1(\tau,\zeta))(\omega) = x_t(\omega, \tau(\omega), \zeta(\omega))$. Therefore, the random variable $\rho_t(\xi)$ is of the form:
\begin{equation*}
	\rho_t(\xi) = \rho^0_t(\xi^0) \ind_{t < \tau} + x_t(\tau,\zeta) \ind_{t \geq \tau}
\end{equation*}
and by Lemma~\ref{lemma:procdecomp}(a) it is $\cG_t$-measurable.

This example is rather artificial, because it requires a financial agent to specify \emph{a priori} an $\bbH$-dynamic risk measure, i.e., a risk measure defined on the initially enlarged filtration. However, it serves as a basis for the next example, which is in our opinion more relevant for applications.
\end{example}

\begin{example}
Let us be given an $\bbF$-dynamic risk measure $\rho^0$ and an indexed family of $\bbF$-dynamic risk measures, i.e., for each $(\theta,e) \in \R^+ \times E$, we have an $\bbF$-dynamic risk measure $\rho^{\theta,e}$. Suppose that for each fixed $t \in [0,T]$ and $x \in \dL^\infty(\cF_T)$ the map $f_t^x \colon \R^+ \times E \to \dL^\infty(\cF_t)$ defined as $f_t^x(\theta, e) \coloneqq \rho_t^{\theta,e}(x)$ is measurable. Then, recalling again (i) of~\citep[Prop. 2.7]{callegaro2013:carthaginian}, we obtain an $\bbH$-dynamic risk measure $\rho^1$ defining:
\begin{equation*}
\rho^1_t(X) \coloneqq \rho^1_t(x(\tau,\zeta)) \coloneqq f_t^{x(\theta,e)}(\theta,e)_{\vert (\theta,e) = (\tau,\zeta)}, \quad t \in [0,T], \, X \in \dL^\infty(\cH_T),
\end{equation*}
where $x$ is an $\cF_T \otimes \cB(\R^+) \otimes \cB(E)$-measurable function, such that $X = x(\tau,\zeta)$.
In fact, thanks to the measurability properties involved, we have that $\rho^1_t(X) \in \dL^\infty(\cH_t)$; moreover, $\rho^1_0(X) \in \R$, since $f_0^x(\theta,e) \in \R$ for any $(\theta,e) \in \R^+ \times E$, given the fact that $\rho^{\theta,e}$ is an $\bbF$-dynamic risk measure; finally, $\rho^1_T(X) = -x(\theta,e)_{\vert (\theta,e) = (\tau,\zeta)} = -x(\tau, \zeta) = -X$.

Finally, we get a $\bbG$-dynamic risk measure proceeding exactly as in the previous example.
This shows that it is possible to define the desired dynamic risk measure by specifying, in addition to a pre-default $\bbF$-dynamic risk measure, a collection of $\bbF$-dynamic risk measures for any possible default time and default mark. This is a feasible and reasonable requirement: in any possible scenario, a financial agent has to specify how she/he intends to update her/his risk measure to capture the riskiness in a financial environment affected by some default event.
\end{example}

\section{\texorpdfstring{Properties of $\bbG$-dynamic risk measures}{Properties of G-dynamic risk measures}} \label{sec: properties-rm}
Let $\rho$ be the dynamic risk measure defined in~\eqref{eq:dynriskmeas}. The aim of this Section is to investigate its properties, that we collect hereunder (see \citep{barrieu2009:pricing,bion-nadal2008,delbaen2006structure,detlefsen-scandolo2005,rosazza:riskgexpect} for more details). These properties can be easily reformulated for $\bbF$- and $\bbH$-dynamic risk measures, for instance whenever referring to the dynamic risk measures $\rho^0$ and $\rho^1$ appearing in the decomposition~\eqref{eq:rhodecomp}. Throughout this Section we assume that the assumptions of Theorems~\ref{th:BSDEJexist} and \ref{th:BSDEJcomp} are in force.

\begin{enumerate}
\item\label{property:01} \emph{Zero-one law:} For all $t \in [0,T]$ and all $A \in \cG_t$:
\begin{equation*}
	\rho_t(\xi \ind_A) = \ind_A \, \rho_t(\xi), \quad \xi \in \dL^\infty(\cG_T).
\end{equation*}
\item\label{property:cashadd} \emph{Translation invariance:} For all $t \in [0,T]$ and all $\eta \in \dL^\infty(\cG_t)$:
\begin{equation*}
	\rho_t(\xi + \eta) = \rho_t(\xi) - \eta, \quad \xi \in \dL^\infty(\cG_T).
\end{equation*}
\item\label{property:hom} \emph{Positive homogeneity:} For all $t \in [0,T]$ and all $\lambda \in \dL^\infty(\cG_t)$, $\lambda \geq 0$:
\begin{equation*}
	\rho_t(\lambda \xi ) = \lambda \, \rho_t(\xi), \quad \xi \in \dL^\infty(\cG_T).
\end{equation*}
\item\label{property:mon} \emph{Monotonicity:} For all $\xi, \eta \in \dL^\infty(\cG_T)$, with $\xi \leq \eta$:
\begin{equation*}
	\rho_t(\xi) \geq \rho_t(\eta), \quad t \in [0,T].
\end{equation*}
\item\label{property:conv} \emph{Convexity:} For all $\xi, \eta \in \dL^\infty(\cG_T)$ and all $\alpha \in [0,1]$:
\begin{equation*}
	\rho_t\bigl(\alpha \xi + (1-\alpha\bigr) \eta) \leq \alpha \rho_t(\xi) + (1-\alpha) \rho_t(\eta), \quad t \in [0,T].
\end{equation*}
\item\label{property:Fatou} \emph{Fatou property:} For any sequence $(\xi_n)_{n \in \N} \subset \dL^\infty(\cG_T)$ and $\xi \in \dL^\infty(\cG_T)$ such that $\xi_n \to \xi$:
\begin{equation*}
	\rho_t(\xi) \leq \liminf_{n \to \infty} \rho_t(\xi_n), \quad t \in [0,T].
\end{equation*}
\item\label{property:timecons} \emph{Time-consistency:} For any $\bbG$-stopping time $\sigma \leq T$, and $\xi \in \dL^\infty(\cG_T)$:
	\begin{equation*}
		\rho_t(\xi) = \rho_t\bigl(-\rho_\sigma(\xi)\bigr), \quad t \leq \sigma.
	\end{equation*}
\end{enumerate}

\begin{proposition}\label{prop:01}
	The dynamic risk measure $\rho$ defined in~\eqref{eq:dynriskmeas} satisfies the following properties:
	\begin{enumerate}
	\item \emph{Zero-one law} if either\footnote{Zero as last argument of the process $g$ stands for the identically null function on $\dB(E)$.} $g(t,0,0,0) = 0, \, \bbP$-a.s., for all $t \in [0,T]$, or both $\rho^0$ and $\rho^1$ satisfy this property.
	\item \emph{Translation invariance} if either $g$ does not depend on $y$ or both $\rho^0$ and $\rho^1$ satisfy this property.
	\item \emph{Positive homogeneity} if either $g$ is positively homogeneous with respect to $(y,z,u)$, $\bbP$-a.s., for all $t \in [0,T]$, or both $\rho^0$ and $\rho^1$ satisfy this property.
	\end{enumerate}
\end{proposition}

\begin{proof}
	We will prove only property~(\ref{property:01}), since the others can be shown with similar arguments.
	
	Let us fix $t \in [0,T], \, A \in \cG_t$ and $\xi \in \dL^\infty(\cG_T)$. We make, first, some preliminary remarks. Since $X \coloneqq \xi \ind_A$ is $\cG_T$-measurable and $\ind_A$ is $\cG_t$-measurable, from Lemma~\ref{lemma:procdecomp} we immediately get:
	\begin{gather*}
		X = X^0 \ind_{T < \tau} + X^1(\tau,\zeta) \ind_{T \geq \tau}, \\
		\ind_A = I^0 \ind_{t < \tau} + I^1(\tau,\zeta) \ind_{t \geq \tau},
	\end{gather*}
	where $X^0, X^1, I^0, I^1$ are appropriately defined as in~\eqref{eq:xidecomp}.
	
	Since $\ind_A \ind_{t < \tau} = I^0\ind_{t < \tau}$, on the event $\{t < \tau\}$ the random variable $I^0$ must take values $0$ or $1$, hence if we define:
	\begin{equation*}
		A^0 \coloneqq \{\omega \in \Omega \colon I^0(\omega) = 1\} \in \cF_t,
	\end{equation*}
	we get $\ind_A \ind_{t < \tau} = \ind_{A^0} \ind_{t < \tau}$. Similarly, defining:
	\begin{equation*}
		A^1 \coloneqq \{\omega \in \Omega \colon I^1(\omega, \tau(\omega), \zeta(\omega)) = 1\} \in \cH_t,
	\end{equation*}
	we get $\ind_A \ind_{t \geq \tau} = \ind_{A^1} \ind_{t \geq \tau}$. Therefore, since $\{T < \tau\} \subset \{t < \tau\}$,
	\begin{equation*}
		\ind_A = \ind_{A^0} \ind_{t < \tau} + \ind_{A^1} \ind_{t \geq \tau} \quad \Longrightarrow \quad \xi \ind_A = \overbrace{\xi^0 \ind_{A^0}}^{X^0} \ind_{T < \tau} + \underbrace{\xi^1(\tau,\zeta)\ind_A}_{X^1(\tau,\zeta)} \ind_{T \geq \tau}.
	\end{equation*}
	
	 If $\rho^0$ and $\rho^1$ both satisfy the zero-one law, then:
	 \begin{gather*}
	 \begin{aligned}
	 	\rho^0_t(X^0) = \rho^0_t(\xi^0 \ind_{A^0}) &= \ind_{A^0} \rho^0_t(\xi^0),
	 	&
	 	\rho^1_t(X^1(\tau,\zeta)) = \rho^1_t(\xi^1(\tau,\zeta) \ind_A) &= \ind_A \rho^1_t(\xi^1(\tau,\zeta)),
	 \end{aligned}
	 \\
	 \Downarrow
	 \\
	 \begin{split}
	 \rho_t(\xi \ind_A) = \rho_t(X) = \rho^0_t(X^0) \ind_{t < \tau} + \rho^1_t(X^1(\tau,\zeta) \ind_{t \geq \tau} = \ind_{A^0} \rho^0_t(\xi^0) \ind_{t < \tau} + \ind_A \rho^1_t(\xi^1(\tau,\zeta))\ind_{t \geq \tau}
	 \\
	 = \ind_A \rho^0_t(\xi^0) \ind_{t < \tau} + \ind_A \rho^1_t(\xi^1(\tau,\zeta))\ind_{t \geq \tau} = \ind_A \rho_t(\xi).
	 \end{split}
	 \end{gather*}
	
	 {\color{blue}I}t remains to prove that the zero-one law holds also if $g(t,0,0,0) = 0, \, \bbP$-a.s., for all $t \in [0,T]$. Now, consider the pairs $(\bar Y^0, \bar Z^0)$ and $(\bar Y^1(\theta,e), \bar Z^1(\theta,e))$, for any $(\theta,e) \in \R^+ \times E$, that are solutions to the Brownian BSDEs~\eqref{eq:BSDE0} and~\eqref{eq:BSDE1}, with terminal conditions $-\xi^0$ and $-\xi^1(\theta,e)$ respectively.
	 Let us define:
	 \begin{align*}
	 	\hat Y^0 &\coloneqq \ind_{A^0} \bar Y^0
	 	&
	 	\hat Z^0 &\coloneqq \ind_{A^0} \bar Z^0
	 	\\
	 	\hat Y^1(\theta,e) &\coloneqq \ind_A \bar Y^1(\theta,e)
	 	&
	 	\hat Z^1(\theta,e) &\coloneqq \ind_A \bar Z^1(\theta,e)
	 \end{align*}
	
	 If $g(s,0,0,0) = 0, \, \bbP$-a.s. for any $s \in [0,T]$, then:
	 \begin{align*}
	 	&\text{On } \{s < \tau\} \quad g^0(s,0,0,0)=0, \\
	 	&\text{On } \{s \geq \tau\} \quad g^1(s,0,0,0,\tau,\zeta)=0.
	 \end{align*}
	 However, we can choose the functions $g^0$ and $g^1$ so that, for any $s \in [0,T]$, $g^0(s,0,0,0) = 0$ and $g^1(s,0,0,0,\tau,\zeta)=0, \, \bbP$-a.s. Indeed, $g$ satisfies the same decomposition as in~\eqref{eq:gdecomp} with the maps $\tilde g^0$ and $\tilde g^1$ defined, for any $s \in [0,T], \, (\theta,e) \in \R^+\times E$, as:
	 \begin{align*}
	 	\tilde g^0(s,y,z,u) &=
	 	\left\{
	 	\begin{aligned}
	 		&g^0(s,y,z,u), & &\text{for } (y,z,u) \neq (0,0,0),\\
	 		&0,	& &\text{otherwise},
	 	\end{aligned}
	 	\right.
	 	\\
	 	\tilde g^1(s,y,z,u,\theta,e) &=
	 	\left\{
	 	\begin{aligned}
	 		&g^1(s,y,z,u,\theta,e), & &\text{for } (y,z,u) \neq (0,0,0),\\
	 		&0,	& &\text{otherwise},
	 	\end{aligned}	 	
	 	\right.
	 \end{align*}
	
	 Hence, there are two possible cases:
	 \begin{itemize}
	 \item On $\{t < \tau\}$ we have that $\rho_t(\xi) = \rho^0_t(\xi^0)$ and $\rho_t(\xi \ind_A) = \rho^0_t(\xi^0 \ind_{A^0})$ from the discussion above. If we compare on this event the terms $\ind_{A^0} \rho^0_t(\xi^0)$ and $\rho^0_t(\xi^0 \ind_{A^0})$ as solutions of their respective BSDEs, we get:
	 \begin{align*}
	 	\ind_{A^0} \rho^0_t(\xi^0) &= -\xi^0\ind_{A^0} + \int_t^T \ind_{A^0} g^0(s, \bar Y^0_s, \bar Z^0_s, \bar Y^1_s(s, \cdot) - \bar Y^0_s) \, \dd s - \int_t^T \ind_{A^0} \bar Z^0_s \, \dd W_s\\
	 	\rho^0_t(\xi^0 \ind_{A^0}) &= -\xi^0\ind_{A^0} + \int_t^T g^0(s, Y^0_s, Z^0_s, Y^1_s(s, \cdot) - Y^0_s) \, \dd s - \int_t^T Z^0_s \, \dd W_s.
	 \end{align*}
	 If we take $Y^0 = \hat Y^0, \, Y^1(\cdot,\cdot) = \hat Y^1(\cdot,\cdot)$ and $Z^0 = \hat Z^0$, then we have that the two right hand sides are a.s. equal. In fact, on the event $A^0$ we have the same equation on both r.h.s., while on $\Omega \setminus A^0$ we have (recall that $\ind_A \ind_{t < \tau} = \ind_{A^0} \ind_{t < \tau}$):
	 \begin{equation*}
	 	0 = \int_t^T g^0(s, Y^0_s, Z^0_s, Y^1_s(s, \cdot) - Y^0_s) \, \dd s - \int_t^T Z^0_s \, \dd W_s = \int_t^T g^0(s, 0, 0, 0) \, \dd s = 0.
	 \end{equation*}
	 So the pair $(\hat Y^0, \hat Z^0)$ is solution to the Brownian BSDE~\eqref{eq:BSDE0} with terminal condition $-\xi^0\ind_{A^0}$, hence $\rho_t^0(\xi^0 \ind_{A^0}) = \ind_{A^0} \rho_t^0(\xi^0)$ and, finally, $\rho_t(\xi \ind_A) = \ind_A \rho_t(\xi)$.
	\item On $\{t \geq \tau\}$ we have that $\rho_t(\xi) = \rho^1_t(\xi^1(\tau,\zeta))$ and $\rho_t(\xi \ind_A) = \rho^1_t(\xi^1(\tau,\zeta) \ind_A)$ from the discussion above. If we compare on this event the terms $\ind_A \rho^1_t(\xi^1(\tau,\zeta))$ and $\rho^1_t(\xi^1(\tau,\zeta) \ind_A)$ as solutions of their respective BSDEs, we get:
	 \begin{align*}
	 	\ind_A \rho^1_t(\xi^1(\tau,\zeta)) &= -\xi^1(\tau,\zeta)\ind_A + \int_t^T \ind_A g^1(s, \bar Y^1_s, \bar Z^1_s, 0, \tau, \zeta) \, \dd s - \int_t^T \ind_A \bar Z^1_s(\tau,\zeta) \, \dd W_s\\
	 	\rho^1_t(\xi^1(\tau,\zeta) \ind_A) &= -\xi^1(\tau,\zeta)\ind_A + \int_t^T g^1(s, Y^1_s, Z^1_s, 0, \tau, \zeta) \, \dd s - \int_t^T Z^1_s(\tau,\zeta) \, \dd W_s.
	 \end{align*}
	 If we take $Y^1(\cdot,\cdot) = \hat Y^1(\cdot,\cdot)$ (as before) and $Z^1(\cdot,\cdot) = \hat Z^1(\cdot,\cdot)$, then we have that the two right hand sides are a.s. equal. In fact, on the event $A^1$ we have the same equation on both r.h.s., while on $\Omega \setminus A^1$ we have:
	 \begin{equation*}
	 	0 = \int_t^T g^1(s, Y^1_s, Z^1_s, 0, \tau, \zeta) \, \dd s - \int_t^T Z^1_s(\tau,\zeta) \, \dd W_s = \int_t^T g^1(s, 0, 0, 0, \tau, \zeta) \, \dd s = 0.
	 \end{equation*}
	 So the pair $(\hat Y^1, \hat Z^1)$ is solution to the Brownian BSDE~\eqref{eq:BSDE1} with terminal condition $-\xi^1(\tau,\zeta)\ind_A$, hence $\rho_t^1(\xi^1(\tau,\zeta) \ind_A) = \ind_A \rho_t^1(\xi^1(\tau,\zeta))$ and, finally, $\rho_t(\xi \ind_A) = \ind_A \rho_t(\xi)$.
	 \end{itemize}
	 Putting together the facts proved on the two disjoint events and recalling that $\ind_{A^0} \ind_{t < \tau} = \ind_A \ind_{t < \tau}$, we get the desired property.
\end{proof}

\begin{remark}
The zero-one law implies that the dynamic risk measure $\rho$ is normalized, i.e., $\rho_t(0) = 0$ for all $t \in [0,T]$. Moreover, if $g$ does not depend on $y$, i.e., $g(t,y,z,u) = g(t,z,u)$ and $g(t,0,0) = 0, \, \bbP$-a.s., for all $t \in [0,T]$, then not only $\rho$ satisfies the translation invariance and zero-one law properties, but is also a \emph{filtration consistent conditional $g$-expectation}, in the sense of~\citep{coquet2002:gexpect,peng1997:BSDEsandgexp}. Recall, however, that in general $\rho^0$ and $\rho^1$ are not induced by $g^0$- and $g^1$-conditional expectations respectively (cfr. Remark~\ref{rem:gexpect}).
\end{remark}

\begin{proposition}\label{prop:conv}
	The dynamic risk measure $\rho$ defined in~\eqref{eq:dynriskmeas} satisfies the following properties:
	\begin{enumerate}\setcounter{enumi}{3}
	\item \emph{Monotonicity}.
	\item \emph{Convexity} if either $g$ is convex with respect to $(y,z,u)$, $\bbP$-a.s., for all $t \in [0,T]$, or both $\rho^0$ and $\rho^1$ satisfy this property.
	\end{enumerate}
\end{proposition}

\begin{proof}
	Property~(\ref{property:mon}) easily follows from Theorem~\ref{th:BSDEJcomp}. To prove property~(\ref{property:conv}), i.e., convexity, let us fix $\xi, \eta \in \dL^\infty(\cG_T), \, \alpha \in [0,1]$ and $t \in [0,T]$. Thanks to Lemma~\ref{lemma:procdecomp} we have that:
	\begin{equation*}
		\alpha \xi + (1-\alpha) \eta = \bigl[\alpha \xi^0 + (1-\alpha)\eta^0\bigr]\ind_{T < \tau} + \bigl[\alpha \xi^1(\tau,\zeta) + (1-\alpha)\eta^1(\tau,\zeta)\bigr]\ind_{T \geq \tau},
	\end{equation*}
	where $\xi^0, \xi^1, \eta^0, \eta^1$ are appropriately defined as in~\eqref{eq:xidecomp}.
	
	If both $\rho^0$ and $\rho^1$ are dynamic convex risk measures then it follows immediately that also $\rho$ is.
	
	Now, consider the pairs $(\bar Y^0, \bar Z^0)$, $(Y^{0,-\xi}, Z^{0,-\xi})$ and $(Y^{0,-\eta}, Z^{0,-\eta})$, that are solutions to the Brownian BSDE~\eqref{eq:BSDE0} with terminal conditions $-[\alpha \xi^0 + (1-\alpha)\eta^0]$, $-\xi^0$ and $-\eta^0$ respectively. Similarly, denote with $(\bar Y^1(\theta,e), \bar Z^1(\theta,e))$, $(Y^{1,-\xi}(\theta,e), Z^{1,-\xi}(\theta,e))$ and $(Y^{1,-\eta}(\theta,e), Z^{1,-\eta}(\theta,e))$ the solutions to the Brownian BSDE~\eqref{eq:BSDE1} with terminal conditions $-[\alpha \xi^1(\theta,e) + (1-\alpha)\eta^1(\theta,e)]$, $-\xi^1(\theta,e)$ and $-\eta^1(\theta,e)$ respectively.
	Let us define:
	\begin{align*}
	 	\hat Y^0 &\coloneqq \alpha Y^{0,-\xi} + (1-\alpha)Y^{0,-\eta}
	 	&
	 	\hat Z^0 &\coloneqq \alpha Z^{0,-\xi} + (1-\alpha)Z^{0,-\eta}
	 	\\
	 	\hat Y^1 &\coloneqq \alpha Y^{1,-\xi} + (1-\alpha)Y^{1,-\eta}
	 	&
	 	\hat Z^1 &\coloneqq\alpha Z^{1,-\xi} + (1-\alpha)Z^{1,-\eta}
	\end{align*}	
	If $g$ is convex with respect to $(y,z,u), \, \bbP$-a.s., for any $s \in [0,T]$, then $g^0$ and $g^1$ are so on $\{s < \tau\}$ and $\{s \geq \tau\}$ respectively. However, differently from the proof of Proposition~\ref{prop:01}, we cannot guarantee that we are able to choose $g^0$ and $g^1$ to be convex $\bbP$-a.s. for any $s \in [0,T]$. Therefore, we have to distinguish three possible cases.
	
	\begin{itemize}
	\item On $\{T < \tau\}$ (hence $t < \tau$), $\rho_t = \rho^0_t$. We have:
	\begin{align*}
	&\alpha \rho^0_t(\xi^0) + (1-\alpha) \rho^0_t(\eta^0)
	\\
	&\quad = -[\alpha \xi^0 + (1-\alpha)\eta^0]
	+\int_t^T \bigl[\alpha g^0(s, Y^{0,-\xi}_s, Z^{0,-\xi}_s, Y^{1,-\xi}_s(s, \cdot) - Y^{0,-\xi}_s)
	\\
	&\qquad \qquad + (1-\alpha) g^0(s, Y^{0,-\eta}_s, Z^{0,-\eta}_s, Y^{1,-\eta}_s(s, \cdot) - Y^{0,-\eta}_s)\bigr] \, \dd s
	\\
	&\qquad -\int_t^T \bigl[\alpha Z_s^{0,-\xi} + (1-\alpha)Z_s^{0,-\eta}\bigr] \, \dd W_s
	\\
	&\quad = -[\alpha \xi^0 + (1-\alpha)\eta^0]
	\\
	&\qquad +\int_t^T \bigl[\lambda^0(s, Y^{0,-\xi}_s, Y^{0,-\eta}_s, Z^{0,-\xi}_s, Z^{0,-\eta}_s, Y^{1,-\xi}_s(s, \cdot) - Y^{0,-\xi}_s, Y^{1,-\eta}_s(s, \cdot) - Y^{0,-\eta}_s, \alpha)
	\\
	&\qquad \qquad + g^0(s, \hat Y^0_s, \hat Z^0_s, \hat Y^1_s(s, \cdot) - \hat Y^0_s)\bigr] \, \dd s
	-\int_t^T \hat Z^0_s \, \dd W_s
	\end{align*}
	where $\lambda^0$ is an almost surely nonnegative process. Hence, since the driver of the last BSDE is greater or equal than $g^0$ almost surely, using Theorem~\ref{th:BSDEJcomp} we get that: $\alpha \rho^0_t(\xi^0) + (1-\alpha) \rho^0_t(\eta^0) \geq \rho^0_t(\alpha \xi^0 + (1-\alpha)\eta^0)$. Finally: $\rho_t(\alpha \xi + (1-\alpha)\eta) \leq \rho_t(\xi) + (1-\alpha) \rho_t(\eta)$.
	\item The same conclusion holds also on $\{t \geq \tau\}$ (hence $T \geq \tau)$ by a similar argument.
	\item On $\{t < \tau \leq T\}, \, \rho_t = \rho^0_t$. We have:
	\begin{align*}
	&\alpha \rho^0_t(\xi^0) + (1-\alpha) \rho^0_t(\eta^0) = \alpha \rho_t(\xi) + (1-\alpha) \rho_t(\eta)
	\\
	&\quad = -[\alpha \xi^0 + (1-\alpha)\eta^0]
	+\int_t^\tau \bigl[\alpha g^0(s, Y^{0,-\xi}_s, Z^{0,-\xi}_s, Y^{1,-\xi}_s(s, \cdot) - Y^{0,-\xi}_s)
	\\
	&\qquad \qquad + (1-\alpha) g^0(s, Y^{0,-\eta}_s, Z^{0,-\eta}_s, Y^{1,-\eta}_s(s, \cdot) - Y^{0,-\eta}_s)\bigr] \, \dd s
	\\
	&\qquad -\int_t^\tau \bigl[\alpha Z_s^{0,-\xi} + (1-\alpha)Z_s^{0,-\eta}\bigr] \, \dd W_s
	\\
	&\qquad
	- \bigr[\alpha \bigl(Y^{1,-\xi}_\tau(\tau, \zeta) - Y^{0,-\xi}_\tau\bigr) + (1-\alpha) \bigl(Y^{1,-\eta}_\tau(\tau, \zeta) - Y^{0,-\eta}_\tau\bigr)\bigr]
	\\
	&\qquad +\int_\tau^T \bigl[\alpha g^1(s, Y^{1,-\xi}_s, Z^{1,-\xi}_s, 0, \tau, \zeta)
	+ (1-\alpha) g^1(s, Y^{1,-\eta}_s, Z^{1,-\eta}_s, 0, \tau, \zeta)\bigr] \, \dd s
	\\
	&\qquad + \int_\tau^T \bigl[\alpha Z_s^{1,-\xi}(\tau, \zeta) + (1-\alpha)Z_s^{1,-\eta}(\tau, \zeta)\bigr] \, \dd W_s
	\\
	&\quad = -[\alpha \xi^0 + (1-\alpha)\eta^0]
	\\
	&\qquad +\int_t^\tau \bigl[\lambda^0(s, Y^{0,-\xi}_s, Y^{0,-\eta}_s, Z^{0,-\xi}_s, Z^{0,-\eta}_s, Y^{1,-\xi}_s(s, \cdot) - Y^{0,-\xi}_s, Y^{1,-\eta}_s(s, \cdot) - Y^{0,-\eta}_s, \alpha)
	\\
	&\qquad \qquad + g^0(s, \hat Y^0_s, \hat Z^0_s, \hat Y^1_s(s, \cdot) - \hat Y^0_s)\bigr] \, \dd s
	-\int_t^\tau \hat Z^0_s \, \dd W_s - (\hat Y^1_\tau(\tau,\zeta) - \hat Y^0_\tau)
	\\
	&\qquad +\int_\tau^T \!\!\bigl[\lambda^1(s, Y^{1,-\xi}_s, Y^{1,-\eta}_s, Z^{1,-\xi}_s, Z^{1,-\eta}_s, \alpha)
	+ g^1(s, \hat Y^1_s, \hat Z^1_s, 0, \tau, \zeta)\bigr] \dd s - \int_\tau^T \hat Z^1_s(\tau, \zeta) \dd W_s
	\\
	&\quad= -[\alpha \xi^0 + (1-\alpha)\eta^0]
	+\int_t^T \bigl[\bar \lambda_s + g^0(s, \hat Y^0_s, \hat Z^0_s, \hat Y^1_s(s, \cdot) - \hat Y^0_s)\bigr] \, \dd s
	-\int_t^T \hat Z^0_s \, \dd W_s
	\end{align*}
	where $\bar \lambda_s$ is an almost surely nonnegative process. Hence, since the driver of the last BSDE is greater or equal than $g^0$ almost surely, using Theorem~\ref{th:BSDEJcomp} we get that: $\alpha \rho^0_t(\xi^0) + (1-\alpha) \rho^0_t(\eta^0) \geq \rho^0_t(\alpha \xi^0 + (1-\alpha)\eta^0)$. Finally:
	\begin{equation*}
		\rho_t(\alpha \xi + (1-\alpha)\eta) \leq \rho_t(\xi) + (1-\alpha) \rho_t(\eta).
	\end{equation*}
	\end{itemize}
	Putting together all the cases analyzed so far, we get that $\rho$ is a dynamic \emph{convex} risk measure.
\end{proof}

\begin{remark}
In general, properties analogous to~(\ref{property:01})-(\ref{property:conv}) for $\rho^0$ and $\rho^1$ cannot be deduced by making assumptions on $g$ alone. This is due to the construction of these dynamic risk measures provided in the proof of Proposition~\ref{prop:rhodecomp} (see also Remark~\ref{rem:gexpect}) and to the fact that any assumption on $g$ is reflected on $g^0$ and $g^1$ only on the events $\{t < \tau\}$ and $\{t \geq \tau\}$ respectively, and not on the whole processes $g^0$ and $g^1$. This should be evident by carefully inspecting the proofs of Propositions~\ref{prop:01} and~\ref{prop:conv}. Clearly this is not an issue, since the risk measures $\rho^0$ and $\rho^1$ appear in the decomposition~\eqref{eq:rhodecomp} on the appropriate events.
\end{remark}

In the following Proposition we prove that the dynamic risk measure $\rho$ satisfies a \emph{continuity property}, that implies the Fatou property.
\begin{proposition}\label{prop:Fatou}
Let the hypotheses of Proposition~\ref{prop:BSDEJcont} hold. Consider the dynamic risk measure $\rho$ defined in~\eqref{eq:dynriskmeas} and $\xi_n, \, \xi \in \dL^\infty(\cG_T), \, n \in \N$. For any $t \in [0,T]$, we have that the following \emph{continuity property} holds:
\begin{equation*}
	\xi^n \to \xi, \quad \text{in } \dL^\infty(\cG_T) \quad \Longrightarrow \quad \rho_t(\xi^n) \to \rho_t(\xi), \quad \text{in } \dL^\infty(\cG_t).
\end{equation*}
Consequently, \eqref{property:Fatou}, i.e., the \emph{Fatou property}, is satisfied as well.
\end{proposition}

%

\begin{proof}
Let us fix $t \in [0,T]$, $\epsilon > 0$ and a sequence $(\xi^n)_{n \in \N} \subset \dL^\infty(\cG_t)$ such that $\xi^n \to \xi \in \dL^\infty(\cG_T)$. We know that there exists a number $n_\epsilon \in \N$ such that $\norm{\xi^n - \xi}_{\dL^\infty} < \epsilon$ for all $n \geq n_\epsilon$, i.e., $|\xi^n - \xi| < \epsilon, \, \bbP$-a.s. Thanks to the decomposition of each of the $\xi^n$ and $\xi$, we can show, first, that there exist $\xi^{0,n}, \, \xi^0$ and $\xi^{1,n}, \, \xi^1$ such that $|\xi^{0,n} - \xi^0| < \epsilon$ and $|\xi^{1,n}(\theta,e) - \xi^1(\theta,e)| < \epsilon$ for all $(\theta,e) \in [0,T] \times E$. Let us write this decomposition as
\begin{equation*}
\begin{gathered}
\xi^n = \tilde\xi^{0,n} \ind_{T < \tau} + \tilde\xi^{1,n}(\tau, \zeta) \ind_{T \geq \tau}, \\
\xi = \tilde\xi^0 \ind_{T < \tau} + \tilde\xi^1(\tau, \zeta) \ind_{T \geq \tau},
\end{gathered}
\end{equation*}
whence
\begin{equation*}
\xi^n - \xi = \bigl[\tilde\xi^{0,n} - \tilde\xi^0\bigr] \ind_{T < \tau} +  \bigl[\tilde\xi^{1,n}(\tau, \zeta) - \tilde\xi^1(\tau,\zeta)\bigr] \ind_{T \geq \tau} \eqqcolon \tilde \delta^{0,n} \ind_{T < \tau} + \tilde \delta^{1,n}(\tau,\zeta) \ind_{T \geq \tau}.
\end{equation*}
Let $K > 0$ be an arbitrary positive constant and define $\xi^{0,n}, \, \xi^0$ and $\xi^{1,n}, \, \xi^1$ as follows:
\begin{equation*}
\begin{gathered}
\xi^{0,n} \coloneqq \tilde \xi^{0,n} \ind_{|\tilde \delta^{0,n}| < \epsilon} + K \ind_{|\tilde \delta^{0,n}| \geq \epsilon},
\\
\xi^0 \coloneqq \tilde \xi^0 \ind_{|\tilde \delta^{0,n}| < \epsilon} + K \ind_{|\tilde \delta^{0,n}| \geq \epsilon},
\\
\xi^{1,n}(\theta,e) \coloneqq \tilde \xi^{1,n}(\theta,e) \ind_{|\tilde \delta^{1,n}(\theta,e)| < \epsilon} + K \ind_{|\tilde \delta^{1,n}(\theta,e)| \geq \epsilon},
\\
\xi^1(\theta,e) \coloneqq \tilde \xi^1(\theta,e) \ind_{|\tilde \delta^{1,n}(\theta,e)| < \epsilon} + K \ind_{|\tilde \delta^{1,n}(\theta,e)| \geq \epsilon},
\end{gathered}
\end{equation*}
for each $(\theta,e) \in [0,T] \times E$. It is clear that $|\xi^{0,n} - \xi^0| < \epsilon$ and $|\xi^{1,n}(\theta,e) - \xi^1(\theta,e)| < \epsilon$ for all $(\theta,e) \in [0,T] \times E$. Moreover, on the event $\{T < \tau\}$ we have that $|\tilde \delta^0| < \epsilon$, hence $\xi^{0,n} = \tilde \xi^{0,n}$ and $\xi^0 = \tilde \xi^0$. Reasoning similarly on the event $\{T \geq \tau\}$ we obtain:
\begin{equation*}
\begin{gathered}
\xi^n = \xi^{0,n} \ind_{T < \tau} + \xi^{1,n}(\tau, \zeta) \ind_{T \geq \tau}, \\
\xi = \xi^0 \ind_{T < \tau} + \xi^1(\tau, \zeta) \ind_{T \geq \tau}.
\end{gathered}
\end{equation*}

This implies that for all $n \geq n_\epsilon$
\begin{equation*}
M < \epsilon \max\bigl\{K^0, \sup_{(\theta,e) \in [0,T]\times E} K^1(\theta,e)\bigr\},
\end{equation*}
where $M$ is the constant defined in Proposition~\ref{prop:BSDEJcont}. From the estimate~\eqref{eq:deltaYestimate} it follows that
\begin{equation*}
|\rho_t(\xi^n) - \rho_t(\xi)| \leq 2M < 2\epsilon \max\bigl\{K^0, \sup_{(\theta,e) \in [0,T]\times E} K^1(\theta,e)\bigr\}
\end{equation*}
and continuity follows immediately. The Fatou property is implied by the continuity property just proved.
\end{proof}

We finally turn our attention to the last property listed at the beginning of this Section, i.e., the time-consistency.

\begin{proposition}\label{prop:timeconsist}
	The dynamic risk measure $\rho$ defined in~\eqref{eq:dynriskmeas} satisfies~(\ref{property:timecons}), i.e., is \emph{time-consistent}.
\end{proposition}

\begin{proof}
	The proof of time-consistency relies on the \emph{flow property} associated to solutions of BSDEs. It can be proved, in an almost identical manner as in the Brownian or in the Brownian plus Poisson cases (see, e.g., \citep{elkaroui1997:BSDEs} and \citep{quenez:BSDEJ}), that the BSDEJ~\eqref{eq:BSDEJ} satisfies it. More specifically, we have that
	\begin{equation*}
		Y_t^{\xi, T} = Y_t^{Y_\sigma^{\xi, T}, \, \sigma}, \quad t \in [0, \sigma], \, \bbP\text{-a.s.,}
	\end{equation*}
	where, for any $\bbG$-stopping time $\sigma' \leq T$ and any $\xi' \in \cG_{\sigma'}$, we denote by $(Y^{\xi',\sigma'}, Z^{\xi',\sigma'}, U^{\xi',\sigma'})$ the solution of the BSDEJ with terminal time $\sigma'$ and terminal condition $\xi'$. To be precise, this solution is extended to the whole interval $[0,T]$ by setting $Y_t^{\xi',\sigma'} = \xi', \, Z_t^{\xi',\sigma'} = 0, \, U_t^{\xi',\sigma'} = 0$ for $t \in (\sigma,T]$. The interested reader is invited to consult the references provided above to get more details on the proof of the flow property.
	
	Thanks to the definition of $\rho$ given in~\eqref{eq:dynriskmeas}, it is immediate to see that the flow property implies
	\begin{equation*}
		\rho_t(\xi) = Y_t^{-\xi,T} = Y_t^{Y_\sigma^{-\xi,T}, \, \sigma} = \rho_t(-\rho_\sigma(\xi)), \quad t \in [0,\sigma], \, \bbP\text{-a.s.},
	\end{equation*}
	i.e., that the dynamic risk measure $\rho$ is time-consistent.
\end{proof}

\begin{remark}
	It is important to notice that the flow property holds also for the Brownian BSDEs~\eqref{eq:BSDE1} and~\eqref{eq:BSDE0}, since we are assuming existence and uniqueness for their respective solutions. More precisely, while the flow property associated to~\eqref{eq:BSDE0} holds for any $\bbF$-stopping time $\sigma^0 \leq T$, the corresponding one associated to~\eqref{eq:BSDE1} holds for any $\bbH$-stopping time $\sigma^1 \in [\tau, T]$.
	
	This fact has interesting consequences on time-consistency of the $\bbF$- and $\bbH$-dynamic risk measures $\rho^0$ and $\rho^1$, given in~\eqref{eq:rhodecomp}. On the one hand, since any $\bbG$-stopping time is also an $\bbH$-stopping time, we get that, in particular, $\rho^1$ satisfy the time-consistency property for any $\bbG$-stopping time $\sigma \in [\tau, T]$. On the other hand, we know from~\citep[Cor. 2.12]{aksamitjeanblanc2017:filtrenl} that for any $\bbG$-stopping time $\sigma$ there exists an $\bbF$-stopping time $\sigma^0$ such that $\sigma \land \tau = \sigma^0 \land \tau$. Therefore, $\rho^0$ satisfies the time-consistency property also for any $\bbG$-stopping time $\sigma \in [0,\tau \land T)$.
\end{remark}

\section{\texorpdfstring{Example{\color{blue}: the entropic risk measure case}}{Example: the entropic risk measure case}}\label{sec:example}
\color{blue}
In this Section we provide an example of $\bbG$-dynamic risk measure. We will, first, explicit its decomposition, as in~\eqref{eq:rhodecomp} and explain its financial meaning. Then, we will provide the evaluation of the riskiness of a simple financial claim and give a numerical simulation.
\normalcolor

We are going to introduce a particular kind of $\bbG$-dynamic \emph{entropic} risk measure, induced by a BSDEJ of the form~\eqref{eq:BSDEJ}. As is known, entropic risk measures are connected to exponential utility functions, i.e, of the form $u(x) = -\gamma \de^{-\frac x\gamma}$, where $\gamma > 0$ is a parameter modeling risk tolerance of the financial agent (see, e.g., \citep{barrieu2009:pricing}). Here, we use this parameter to introduce a dependence of the risk aversion of the agent (the inverse of $\gamma$) on the possible default event. We will assume, in particular, that whenever a default event occurs, the agent becomes more risk averse,
\color{blue}
and we will quantify \emph{how much} she/he changes her/his preferences through a specific function.

To carry our program, let us simplify our setting, first. We are given a one-dimensional Brownian motion $W = (W_t)_{t \geq 0}$ and a random variable $\tau$ with values in $\R^+$, describing a \emph{default time}. As usual, we denote by $\bbF = (\cF_t)_{t \geq 0}$ the completed natural filtration generated by $W$ and call it the \emph{reference filtration}.
We define for all $t \geq 0$ the following families of $\sigma$-algebras:
\begin{itemize}
	\item $\cD_{t} \coloneqq \sigma(\ind_{\tau \leq s}, \, 0 \leq s \leq t)$;
	\item $\cG_t \coloneqq \bigcap_{s > t} \bigl(\cF_s \lor \cD_s\bigr)$;
\end{itemize}
and we call $\bbG \coloneqq (\cG_t)_{t \geq 0}$ the \emph{progressively enlarged filtration}.
We assume that $\tau$ is the first jump time of a Cox process (see, e.g., \citep{aksamitjeanblanc2017:filtrenl, bieleckirutkowski2002:creditrisk}), with $\bbF$-stochastic intensity $\lambda^0 = (\lambda^0_t)_{t \geq 0}$ of the form $\lambda^0_t = \nu(L_t), \, t \geq 0$, where $\nu(x) = \de^{-x}$ and $L = (L_t)_{t \geq 0}$ satisfies the SDE
\begin{equation*}
\dd L_t = \mu L_t \, \dd t + \sigma L_t \, \dd W_t,
\end{equation*}
with $\mu, \sigma > 0$ fixed parameters.
More precisely, we suppose that the probability space $(\Omega, \cA, \bbP)$ is rich enough to support a random variable $\eta$, independent of $\bbF$, having exponential distribution and mean $1$, and we define:
\begin{equation*}
\tau \coloneqq \inf \biggl\{t \in \R^+ \colon \int_0^t \nu(L_s) \, \dd s \geq \eta \biggr\}.
\end{equation*}
Notice that this construction of the default time $\tau$, sometimes called \emph{canonical}, guarantees that both the density hypothesis (Assumption~\ref{hp:density}) and the immersion hypothesis (Assumption~\ref{hp:martingale}) are satisfied (see \citep[Lemma 2.25 and Lemma 2.28]{aksamitjeanblanc2017:filtrenl}).
This intensity-based way of modeling a default time is called \emph{reduced-form approach} in credit risk modeling.
Moreover, in this way the $\bbG$-\emph{compensator} $\lambda_t \, \dd t$ of the random counting measure $\mu$ has $\bbG$-\emph{stochastic intensity} $\lambda = (\lambda_t)_{t \geq 0}$ given by
\begin{equation*}
\lambda_t = \lambda^0_t \ind_{t < \tau},
\end{equation*}
hence, given our construction, $\lambda$ is a bounded process.

We suppose that the driver $g$ of the BSDEJ~\eqref{eq:BSDEJ} does not depend on $(y,u)$ and is of the form $g(\omega, t, z) = \frac 12 z^2 f(t, \tau(\omega))$, where
\begin{equation*}
f(t, \theta) =
\left\{
\begin{aligned}
&1,	&\text{if } t \leq \theta,\\
&\frac{1}{\gamma(\theta)},	&\text{if } t > \theta,
\end{aligned}
\right.
\end{equation*}
and $\gamma \colon \R^+ \to (0,1)$ is a measurable function, that gives the risk tolerance parameter after a default event has occurred. Notice that its possible values are lower than the value of the risk tolerance prior to default, i.e., one (which is, of course, just a conventional value).

It is clear that this driver is $\cP(\bbG) \otimes \cB(\R)$-measurable, hence it admits the decomposition $g(t,z) = g^0(z) \ind_{t \leq \tau} + g^1(z, \tau) \ind_{t > \tau}$, where
\begin{align*}
g^0(z) &=  \frac 12 z^2, &
g^1(z,\theta) &= \frac{1}{2\gamma(\theta)} z^2.
\end{align*}
Moreover, from Proposition~\ref{prop:BSDEJquadexist} and Theorem~\ref{th:BSDEJquaduniq}, we have existence and uniqueness
of the solution $(Y,Z,U)$ to the BSDEJ~\eqref{eq:BSDEJ}. In particular, this entails that there exist unique solutions $(Y^0, Z^0)$, and, for each $\theta \in \R^+, \, \bigl(Y^1(\theta), Z^1(\theta)\bigr)$, to the Brownian BSDEs:
\begin{gather*}
Y_t^0 = -\xi^0 + \int_t^T g^0(Z_s^0) \, \dd s - \int_t^T Z_s^0 \, \dd W_s, \quad 0 \leq t \leq T,
\\
Y_t^1(\theta) = -\xi^1(\theta) + \int_t^T g^1(Z_s^1(\theta), \theta) \, \dd s
- \int_t^T Z_s^1(\theta) \, \dd W_s, \quad \theta \land T \leq t \leq T.
\end{gather*}
More explicitly, we get
\begin{gather}
Y_t^0 = \log \bbE[\de^{-\xi^0} \mid \cF_t], \quad 0 \leq t \leq T, \label{eq:Y0entr} \\
Y_t^1(\theta) = \gamma(\theta) \log \bbE[\de^{-\frac{\xi^1(\theta)}{\gamma(\theta)}} \mid \cF_t], \quad \theta \land T \leq t \leq T. \label{eq:Y1entr}
\end{gather}
Therefore, if we define the $\bbG$-dynamic risk measure $\rho_t(\xi) \coloneqq Y_t^{- \xi}, \, t \in [0,T], \, \xi \in \dL^\infty(\cG_T)$, from Proposition~\ref{prop:rhodecomp} we have that, for all $t \in [0,T]$
\begin{align*}
	\rho_t(\xi) &= \rho^0_t(\xi^0) \, \ind_{t < \tau} + \rho^1_t\bigl(\xi^1(\tau)\bigr) \, \ind_{t \geq \tau}, \\
	\rho^0_t(\xi^0) &= Y^0_t, \\
	\rho^1_t(\xi^1(\tau)) &= Y^1_t(\tau), \quad \text{on } \{t \geq \tau\},
\end{align*}

\normalcolor
Notice that $\rho$ is a filtration consistent (i.e., the zero-one law holds), translation invariant, monotone, convex and time-consistent $\bbG$-dynamic risk measure. It is immediate to see that the same properties hold for the $\bbF$-dynamic risk measure $\rho^0$ and, if properly defined on the event $\{t < \tau\}$, also for the $\bbH$-dynamic risk measure $\rho^1$.
Thanks to the results contained in~\citep{kobylanski2000:BSDEs} on quadratic BSDEs, one can prove that $\rho$ satisfies also the Fatou property.

The $\bbG$-dynamic risk measure that we considered so far has an intuitive financial meaning: if no default event occurs, the financial agent uses a reference entropic risk measure, in this case $\rho^0$. If a default occurs, she/he uses the entropic risk measure $\rho^1$ that reflects her/his changes in preferences, due to an increased risk aversion. This updating feature has interesting connections with backward and forward utilities (see \citep{musiela2007investment}) and dynamic stochastic utilities (or, better, with the corresponding conditional certainty equivalent - see \citep{frittelli-maggis2011}) where utility may change with state and time.

\color{blue}
To illustrate better this financial interpretation, suppose that a financial market is modeled under the framework described above and that a financial agent stipulates the following contract expiring at time $T$:
\begin{itemize}
\item The price at time $t=0$ of the contract is one monetary unit.
\item If there is no default during the contract, at the terminal time $T$ the agent receives the amount of money initially invested plus the value $W_T$ of the Brownian motion.
\item If there is default during the contract, at the terminal time $T$ the agent receives only the fraction $\displaystyle W_T \frac{T-\tau}{T}$ of the value $W_T$ depending on the time to maturity $T-\tau$ when the default $\tau$ occurs.
\end{itemize}
The financial position representing the profit and loss of the contract is, thus, described by the following random variable:
\begin{equation*}
\xi = W_T \, \ind_{T < \tau} + \Bigl(W_T\frac{T-\tau}{T} - 1\Bigr) \, \ind_{T \geq \tau}.
\end{equation*}
The proposed structure of the claim permits us, as we will see shortly, to get an explicit formula to evaluate its riskiness through the entropic risk measure introduced earlier.

It is clear that $\xi$ is $\cG_T$-measurable and that $\xi = \xi^0 \ind_{T < \tau} + \xi^1(\tau) \ind_{T \geq \tau}$, where $\xi^0 = W_T$ and $\xi^1(\theta) = W_T\frac{T-\theta}{T} - 1$, for all $0 \leq \theta \leq T$. From~\eqref{eq:Y0entr}-\eqref{eq:Y1entr} we get, after routine computations,
\begin{gather*}
\rho^0_t(\xi^0) = \frac{T-t}{2} - W_t, \\
\rho^1_t(\xi^1(\tau)) = 1 + \frac{(T-\tau)^2}{2\gamma(\tau)T^2}(T-t) - \frac{T-\tau}{T} W_t.
\end{gather*}
These explicit formulas allow us to simulate numerically the evaluation of the riskiness of this financial position and to highlight the relevance of our result in the context of financial markets where defaults may occur.

\begin{remark}
Before giving our numerical results, let us stress the fact that under the frameworks studied so far in the literature for dynamic risk measures induced by BSDEs, a claim such as the one proposed in this example cannot be evaluated. In fact, the random variable $\xi$ is expressed as a measurable function of an $\cF_T$-measurable random variable, $W_T$ in this case, and of the random time $\tau$, which is not independent of $W_T$. The last fact is crucial and prevents to use any known result, to the best of our knowledge. At the same time, it is clear that an investigation of dynamic risk measures in this context or, more generally, under the theory of progressive enlargement of filtration, should be considered with attention since in this framework many interesting results can be deduced and more refined financial market models can be studied.
\end{remark}

We provide below five simulated trajectories of the dynamic risk measure $\rho$ induced by the BSDEJ~\eqref{eq:BSDEJ} and of its components $\rho^0$ and $\rho^1$ (before the default time, the value of $\rho^1$ is conventionally fixed at zero). The data used are the following: $T=1$, $\mu=1$, $\sigma=0.1$, $\gamma(\theta) = 1-0.9 \, \de^{-\theta}$, with $1\,000$ time grid discretization steps. Notice, in particular, that the proposed $\gamma$ function describes a situation where, in case of default, the risk-tolerance parameter gets closer to the default-free one (equal to $1$, in this example), as the default time moves further away. The simulated default times are given under the picture.

\begin{figure}[H]
\includegraphics[scale=1]{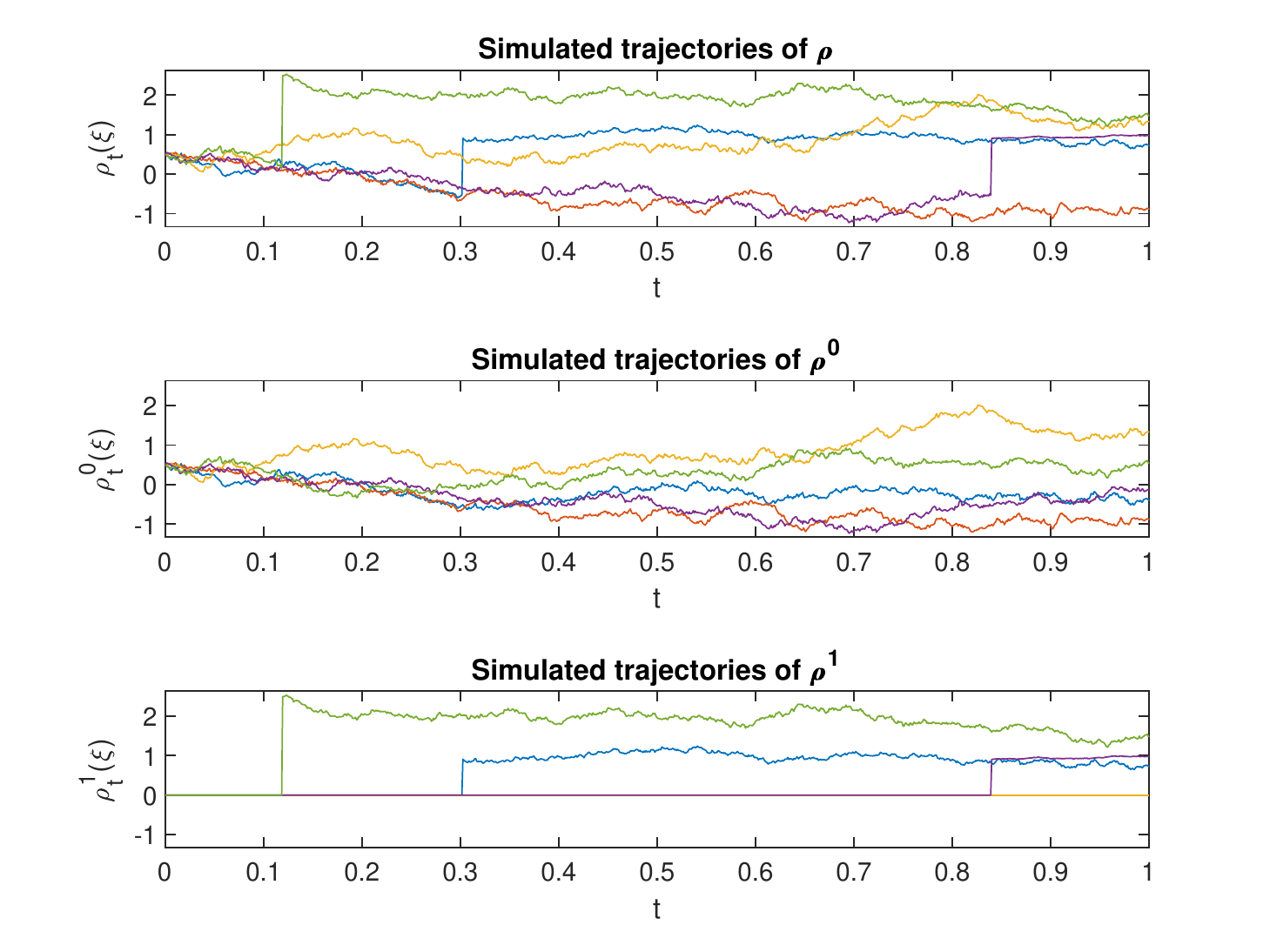}
\caption{Simulated trajectories of $\rho, \, \rho^0$ and $\rho^1$ corresponding to the evaluation of the financial claim $\xi = W_T \, \ind_{T < \tau} + \Bigl(W_T\frac{T-\tau}{T} - 1\Bigr) \, \ind_{T \geq \tau}$. The simulated default times are: $\tau = 0.1191$ (green trajectory), $\tau = 0.3023$ (blue trajectory), $\tau = 0.8398$ (purple trajectory), $\tau = 1.4381$, $\tau = 2.0087$.}
\end{figure}

It is easily seen that to have a tool to change, in a time-consistent way, the evaluation of riskiness in case of a default event is a crucial feature if risk measures are to be applied, for instance, in credit markets: there are scenarios under which the financial position $\xi$ is acceptable at least for a period of time before default and then becomes not acceptable afterwards. This behavior is due not only to the dependence of $\xi$ on the default time but also on the updating of the risk-tolerance parameter of the financial agent.

A second and final example provides an evaluation of the effect of the sole updating feature of the risk measure $\rho$. Suppose that the net financial position under evaluation is $\xi = W_T$. This is clearly $\cF_T$-measurable, hence the effect of the random event determined by $\tau$ (be it a default or a sudden new information, for instance) is tangible only through the risk-tolerance parameter of the financial agent.
In fact, in this case the evolution of the risk evaluations $\rho^0(\xi^0)$ and $\rho^1(\xi^1(\tau))$ is given by
\begin{gather*}
\rho^0_t(\xi^0) = \frac{T-t}{2} - W_t, \\
\rho^1_t(\xi^1(\tau)) = \frac{T-t}{2\gamma(\tau)} - W_t.
\end{gather*}

As before, we provide five simulated trajectories of the dynamic risk measure $\rho$ induced by the BSDEJ~\eqref{eq:BSDEJ} and of its components $\rho^0$ and $\rho^1$ (where, as in the previous example, the value of $\rho^1$ before the default time is conventionally fixed at zero). Data are equal to those of the previous simulations and the simulated default times are given under the picture.

\begin{figure}[H]
\includegraphics[scale=1]{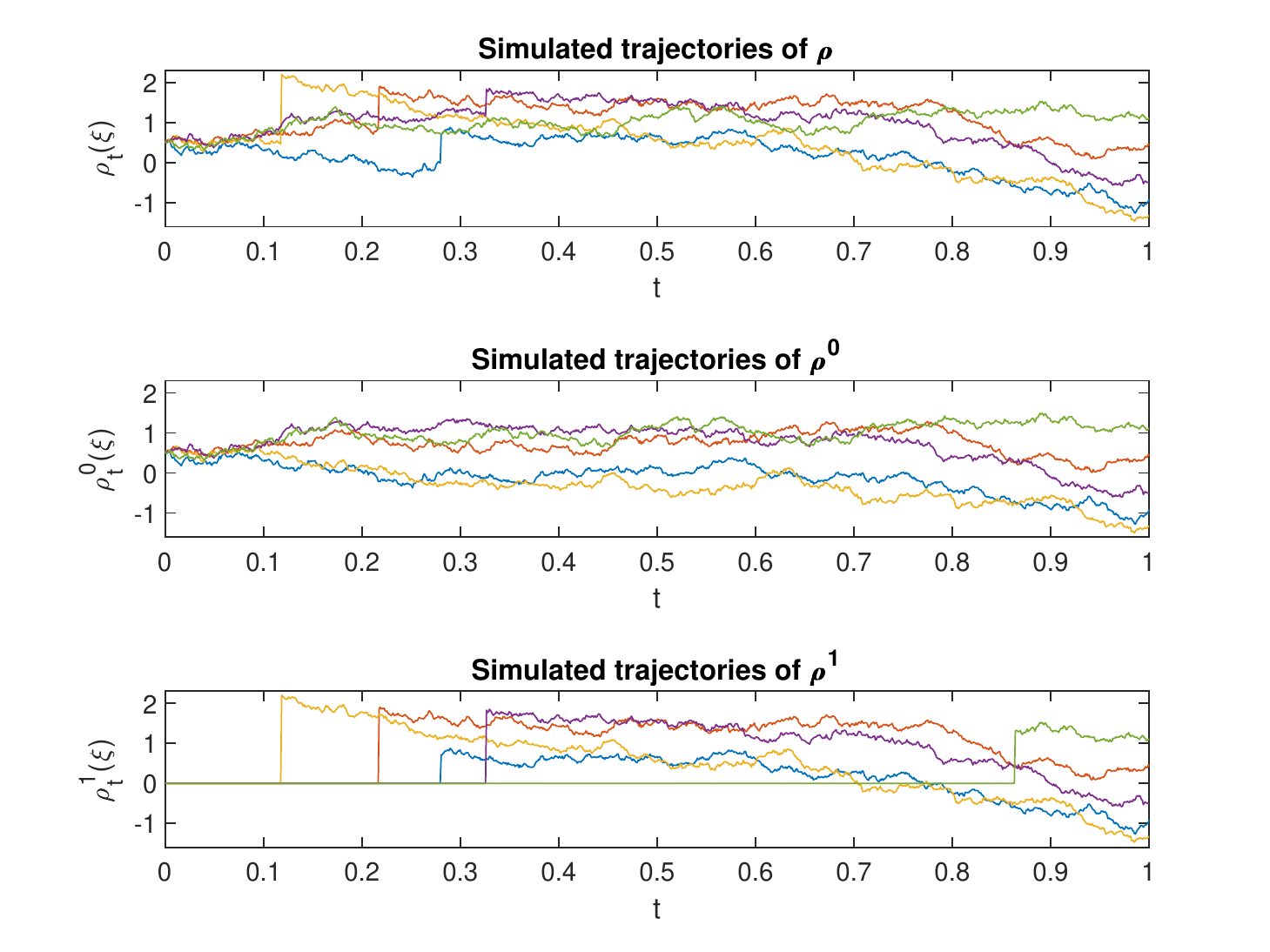}
\caption{Simulated trajectories of $\rho, \, \rho^0$ and $\rho^1$ corresponding to the evaluation of the financial claim $\xi = W_T$. The simulated default times are: $\tau = 0.1181$ (yellow trajectory), $\tau = 0.2172$ (red trajectory), $\tau = 0.2803$ (blue trajectory), $\tau = 0.3263$ (purple trajectory), $\tau = 0.8639$ (green trajectory).}
\end{figure}
Also in this case, it is plain to see that allowing the possibility of changing the risk-tolerance parameter upon the arrival of a default event or new information radically changes (as one expects) the evaluation of financial claims, even if they do not depend on the event itself. It can be also seen that this change fades away with time since, as we expect from theory, $\rho^0_T(\xi) = \rho^1_T(\xi)$ if $\xi$ is $\cF_T$-measurable.

\begin{remark}
The choice of $g^0, g^1$ in this example (and, correspondingly, the choice of $\rho^0, \rho^1$) and their financial interpretation and motivation, allows us to point out an interesting comparison between agents that react differently to the arrival of new information brought by the random time $\tau$.
\normalcolor
Let us consider a terminal condition $\xi \in \dL^\infty(\cF_T)$ and fix a time $t \in [0,T]$. For each $(\theta,e) \in [0,t) \times E$, we have that $g^0(z) \leq g^1(z, \theta, e)$ for any $z$ and, consequently, by the Comparison Theorem on Brownian BSDEs, $Y_t^{0 , - \xi} \leq Y_t^{1, - \xi}(\theta,e)$.
Reasoning as in the proof of Proposition~\ref{prop:BSDEJcont}, this implies that $Y_t^{0 , - \xi} \leq Y_t^{1, - \xi}(\tau,\zeta)$, on the event $t > \tau$, i.e., after the default. In turn, this entails that $\rho^0_t(\xi) \leq \rho^1_t(\xi)$, on the event $t > \tau$.

Suppose, now, to compare two distinct financial agents, A and B: both have access to default-related information but, while A updates her/his preferences according to the change of risk tolerance parameter described in this example, B sticks with the initial value of $\gamma$, thereby not modifying her/his preferences.
Following the above reasoning, at any time $t$ after the default A will adopt a more conservative risk measure than B, i.e., $\rho^1$ versus $\rho^0$, thus considering the default event as something adding riskiness to her/his financial position $\xi$. Notice that to perform this comparison, one must assume that $\xi$ is a financial claim that is not influenced by the default event, i.e., $\xi$ is $\cF_T$-measurable.

A different and opposite argument can be also provided by considering an alternative point of view. In fact, a more conservative \emph{reference} risk measure $\rho^0$ can be chosen in order to stress and penalize the uncertainty before default. In this case, the default event is interpreted as something adding information to the knowledge of the financial agent and, thus, reducing risk.
\end{remark}

\normalcolor

\normalcolor
\section{Some results on the dual representation of the induced risk measure} \label{sec: representation-rm}
The purpose of the last Section of the paper is to study the robust representation of the $\bbG$-dynamic risk measure $\rho$ introduced in~\eqref{eq:dynriskmeas}. We are able to give some results in this direction that only partially answer to the following question: is it possible to prove that the penalty term appearing in the dual representation of $\rho$ admits a decomposition similar to~\eqref{eq:rhodecomp}?

Without any additional assumption, we are not able to answer definitively to that question. However, with further hypotheses, we can provide a partial affirmative answer in the form of an inequality, as proved in Proposition~\ref{prop:alphadecomp}.

Throughout this Section, we assume that the hypotheses of Theorems~\ref{th:BSDEJexist} and~\ref{th:BSDEJcomp} are satisfied, together with properties~(\ref{property:01}), (\ref{property:cashadd}), (\ref{property:conv}) and~(\ref{property:Fatou}). We recall, also, that thanks to Theorem~\ref{th:BSDEJcomp}, the risk measure $\rho$ satisfies~(\ref{property:mon}), i.e., is monotone.

To begin with, we notice that the Fatou property entails the following robust representation for any $t \in [0,T]$ (see, e.g., \citep[Th. 11.2]{follmerschied2011:stochfin}):
\begin{gather}\label{eq:rhorobustrepr}
	\rho_t(\xi) = \esssup_{\bbQ \in \cQ_t} \bigl\{\bbE_\bbQ[-\xi \, | \, \cG_t] - \alpha_t(\bbQ) \bigr\}, \quad \xi \in \dL^{\infty}(\cG_T),
\end{gather}
where $\cQ_t = \{\bbQ \in \cM_1(\cG_T) \colon \bbQ \ll \bbP_{\vert \cG_T}, \, \bbQ_{\vert \cG_t} = \bbP\}$ and $\alpha_t$ is the \emph{penalty term}:
\begin{equation}\label{eq:penalty}
	\alpha_t(\bbQ) = \esssup_{\xi \in \dL^\infty(\cG_T)} \bigl\{\bbE_\bbQ[-\xi \, | \, \cG_t] - \rho_t(\xi)\bigr\} = \esssup_{\xi \in \dL^\infty(\cG_T), \, \rho_t(\xi) \leq 0} \bbE_\bbQ[-\xi \, | \, \cG_t], \quad \bbQ \in \cQ_t.
\end{equation}

We can prove, first, that $\rho$ satisfies another property, guaranteeing that the set of probability measures considered in~\eqref{eq:rhorobustrepr} does not change in time. Since $\rho$ is time consistent, this is true (see, e.g., \citep[Prop. 4.4]{follmer2006:convexrisk}) if
\begin{equation}\label{eq:relevance}
	\forall \xi \in \dL^\infty(\cG_T), \, \xi \geq 0, \, \exists \lambda > 0 \, \text{ s.t. } \rho_0(-\lambda \xi) > \rho_0(0).
\end{equation}

\begin{proposition}\label{prop:relevance}
	The risk measure $\rho_0 \colon \dL^\infty(\cG_T) \to \R$ defined in~\eqref{eq:dynriskmeas} satisfies~\eqref{eq:relevance}, i.e., the relevance property.
\end{proposition}

\begin{proof}
Let us fix $\xi \in \dL^\infty(\cG_T)$, with $\xi \geq 0$ (excluding, clearly, the case $\xi = 0, \, \bbP$-a.s., for which~\eqref{eq:relevance} is never satisfied).
First of all, let us notice that $\rho_0(0) = 0$, thanks to the zero-one law. What we need to show, then, is that there exists $\lambda >0 $ such that:
\begin{equation*}
	\rho_0(-\lambda \xi) = Y_0 > 0,
\end{equation*}
where $Y$ is the first component of the solution to the BSDEJ~\eqref{eq:BSDEJ} with terminal condition $\lambda \xi$.
If we denote with $(Y^0, Z^0)$ the solution to the BSDE~\eqref{eq:BSDE0} with terminal condition $\lambda \xi^0$ (where $\xi^0$ is the random variable appearing in the decomposition~\eqref{eq:xidecomp}), since $\rho_0 = \rho^0_0$, we have to prove that $Y_0^0 > 0$ for some $\lambda > 0$,  We have that:
\begin{equation*}
	Y_0^0 = \lambda \xi^0 + \int_0^T g^0(s, Y^0_s, Z^0_s, Y^1_s(s, \cdot) - Y^0_s) \, \dd s - \int_0^T Z^0_s \, \dd W_s.
\end{equation*}
Taking expectations on both sides yields:
\begin{equation*}
	Y_0^0 = \lambda \bbE[\xi^0] + \bbE\biggl[\int_0^T g^0(s, Y^0_s, Z^0_s, Y^1_s(s, \cdot) - Y^0_s) \, \dd s\biggr].
\end{equation*}
Since $\xi^0$ is not zero $\bbP$-a.s., we have that $\bbE[\xi^0] > 0$. Moreover, since we are assuming existence and uniqueness of the solution to the Brownian BSDE, we have that the second expectation exists and is finite. Hence, if $\bbE\bigl[\int_0^T g^0(s, Y^0_s, Z^0_s, Y^1_s(s, \cdot) - Y^0_s) \, \dd s\bigr] \geq 0$, any $\lambda > 0$ guarantees that $Y^0_0 > 0$; otherwise, we can choose:
$\lambda > - \dfrac{1}{\bbE[\xi^0]} \bbE\bigl[\int_0^T g^0(s, Y^0_s, Z^0_s, Y^1_s(s, \cdot) - Y^0_s) \, \dd s\bigr]$. \qedhere
\end{proof}

Thanks to Proposition~\ref{prop:relevance} we obtain the following robust representation for $\rho$:
\begin{equation}\label{eq:rhorobustreprrel}
	\rho_t(\xi) = \esssup_{\bbQ \in \cQ} \bigl\{\bbE_\bbQ[-\xi \, | \, \cG_t] - \alpha_t(\bbQ) \bigr\}, \quad \xi \in \dL^{\infty}(\cG_T), \, t \in [0,T],
\end{equation}
where $\cQ \coloneqq \{\bbQ \in \cM_1(\cG_T) \colon \bbQ \sim \bbP_{\vert \cG_T}\}$.
We proceed, now, giving a decomposition of the conditional expectation appearing in~\eqref{eq:rhorobustreprrel}.


\begin{proposition}\label{prop:expectdecomp}
{\color{blue}
Let $\bbQ \in \cQ$ and define the $\cG_T$-measurable random variable $L \coloneqq \frac{\dd \bbQ}{\dd \bbP_{\vert \cG_T}}$, admitting decomposition $L = L^0 \ind_{T < \tau} + L^1(\tau,\zeta) \ind_{T \geq \tau}$. Then, for any $\xi \in \dL^\infty(\cG_T)$}
\begin{equation}\label{eq:Qexpectdecomp}
\bbE_\bbQ[-\xi \mid \cG_t] = \Phi^0 \ind_{t < \tau} + \Phi^1(\tau,\zeta) \ind_{t \geq \tau},
\end{equation}
where
\begin{equation}\label{eq:Phi01}
\begin{gathered}
\Phi^0 \coloneqq \frac{1}{Z_t^\bbP} \bbE\biggl[\int_0^{+\infty} \int_E \bigl\{ -\xi^0 L^0 \ind_{(T, +\infty)}(\theta) -\xi^1(\theta,e) L^1(\theta,e) \ind_{(t,T]}(\theta)\bigr\} \gamma_T(\theta, e) \, \dd \theta \, \dd e \mid \cF_t\biggr],
\\
\Phi^1(\theta,e) \coloneqq \frac{\bbE[-\xi^1(\theta,e) L^1(\theta,e) \gamma_T(\theta,e)\mid \cF_t]}{\gamma_t(\theta,e)},
\end{gathered}
\end{equation}
and $Z_t^\bbP \coloneqq \bbP(t < \tau \mid \cF_t)$ is the Az\'ema supermartingale associated with $\tau$.
\end{proposition}

\begin{proof}
Fix $\bbQ \in \cQ$ and $\xi \in \dL^\infty(\cG_T)$ and denote $L \coloneqq \frac{\dd \bbQ}{\dd \bbP_{\vert \cG_T}}$. Since the random variables $\xi$ and $L$ are $\cG_T$-measurable and $\bbE_\bbQ[-\xi \mid \cG_t]$ is $\cG_t$-measurable, we obtain from Lemma~\ref{lemma:procdecomp} the following decompositions
\begin{gather*}
\xi = \xi^0 \ind_{T < \tau} + \xi^1(\tau,\zeta) \ind_{T \geq \tau},\\
L = L^0 \ind_{T < \tau} + L^1(\tau,\zeta) \ind_{T \geq \tau},\\
\bbE_\bbQ[-\xi \mid \cG_t] = \Phi^0 \ind_{t < \tau} + \Phi^1(\tau,\zeta) \ind_{t \geq \tau},
\end{gather*}
where $\xi^0, \xi^1, L^0, L^1, \Phi^0, \Phi^1$ are appropriately defined as in~\eqref{eq:xidecomp}. We aim at identifying the last two objects.

On $\{t < \tau\}$ we have that
\begin{equation*}
\ind_{t < \tau} \bbE_\bbQ[-\xi \mid \cG_t] = \bbE_\bbQ[-\xi \ind_{t < \tau} \mid \cG_t] = \bbE[-(\xi^0 L^0 \ind_{T < \tau} + \xi^1(\tau,\zeta) L^1(\tau,\zeta) \ind_{T \geq \tau}) \ind_{t < \tau} \mid \cG_t].
\end{equation*}
Since the random variable inside the brackets is integrable, we get from~\citep[Lemma 2.9]{aksamitjeanblanc2017:filtrenl}
\begin{equation*}
\ind_{t < \tau} \bbE_\bbQ[-\xi \mid \cG_t] = \ind_{t < \tau} \frac{1}{Z_t^\bbP} \bbE[-(\xi^0 L^0 \ind_{T < \tau} + \xi^1(\tau,\zeta) L^1(\tau,\zeta) \ind_{T \geq \tau}) \ind_{t < \tau} \mid \cF_t].
\end{equation*}
Noticing that $\{T < \tau\} \subset \{t < \tau\}$ and using Assumption~\ref{hp:density} we get
\begin{align*}
&\mathrel{\phantom{=}} \bbE[-(\xi^0 L^0 \ind_{T < \tau} + \xi^1(\tau,\zeta) L^1(\tau,\zeta) \ind_{T \geq \tau}) \ind_{t < \tau} \mid \cF_t] \\
&= \bbE\bigl[\bbE[-(\xi^0 L^0 \ind_{T < \tau} + \xi^1(\tau,\zeta) L^1(\tau,\zeta) \ind_{t < \tau \leq T}) \mid \cF_T ] \mid \cF_t\bigr] \\
&= \bbE\biggl[\int_0^{+\infty} \int_E \bigl\{ -\xi^0 L^0 \ind_{(T, +\infty)}(\theta) -\xi^1(\theta,e) L^1(\theta,e) \ind_{(t,T]}(\theta)\bigr\} \gamma_T(\theta, e) \, \dd \theta \, \dd e \mid \cF_t\biggr]
\end{align*}
whence the definition of $\Phi^0$.

On $\{t \geq \tau\}$, since $\{t \geq \tau\}$ implies $\{T \geq \tau\}$ and is incompatible with $\{T < \tau\}$, we obtain
\begin{equation*}
\ind_{t < \tau} \bbE_\bbQ[-\xi \mid \cG_t] = \bbE_\bbQ[-\xi \ind_{t < \tau} \mid \cG_t] = \bbE[-\xi^1(\tau,\zeta) L^1(\tau,\zeta) \mid \cG_t] \, \ind_{t \geq \tau}.
\end{equation*}
Since $-\xi^1(\tau,\zeta) L^1(\tau,\zeta)$ is an $\cH_T$-measurable and $\bbP$-integrable random variable, from \citep[Lemma 2.10]{callegaro2013:carthaginian} we get
\begin{equation*}
\bbE[-\xi^1(\tau,\zeta) L^1(\tau,\zeta) \mid \cG_t] \, \ind_{t \geq \tau} =
\frac{\bbE[-\xi^1(\theta,e) L^1(\theta,e) \gamma_T(\theta,e)\mid \cF_t]}{\gamma_t(\theta,e)} \, \ind_{t \geq \tau}
\end{equation*}
whence the definition of $\Phi^1$.
\end{proof}

\begin{remark}\label{rem:Phi1}
Thanks to \citep[Eq. (2.2)]{callegaro2013:carthaginian}, we can write the second summand of~\eqref{eq:Qexpectdecomp} in a different way. In fact:
\begin{equation*}
\Phi^1(\tau,\zeta) = \frac{\bbE[-\xi^1(\theta,e) L^1(\theta,e) \gamma_T(\theta,e)\mid \cF_t]_{\vert (\theta,e) = (\tau,\zeta)}}{\gamma_t(\tau,\zeta)} = \bbE[-\xi^1(\tau,\zeta) L^1(\tau,\zeta) \mid \cH_t].
\end{equation*}
Therefore
\begin{align*}
&\mathrel{\phantom{=}} \ind_{t \geq \tau} \bbE_\bbQ[-\xi \mid \cG_t] = \ind_{t \geq \tau} \bbE[-\xi^1(\tau,\zeta) L^1(\tau,\zeta) \mid \cH_t]
\\
&= \bbE[ -\xi^1(\tau,\zeta) \ind_{t \geq \tau} \bigl(L^0 \ind_{T < \tau} + L^1(\tau,\zeta) \ind_{T \geq \tau} \bigr) \mid \cH_t] = \bbE[ -\xi^1(\tau,\zeta) L \ind_{t \geq \tau}\mid \cH_t]
\\
&= \bbE_{\bbQ^1}[-\xi^1(\tau,\zeta) \mid \cH_t] \, \ind_{t \geq \tau},
\end{align*}
where $\bbQ^1$ is the probability measure on $(\Omega, \cH_T)$ defined by $\dd \bbQ^1 = L \dd \bbP_{\vert \cH_T}$.
\end{remark}

To prove the main (and final) result of this Section, i.e., an inequality satisfied by the penalty term appearing in~\eqref{eq:penalty}, we need to properly restrict the class of probability measures $\cQ$.
Let us define $\cQ^{\hookrightarrow}$ the subset of $\cQ$ containing all probability measures $\bbQ \in \cQ$ satisfying the immersion property, i.e., such that $\bbF$ is immersed in $\bbG$ under $\bbQ$.

\begin{theorem}\label{prop:alphadecomp}
For any $t \in [0,T]$ and any $\bbQ \in \cQ^{\hookrightarrow}$ the following inequality holds for the penalty term given by~\eqref{eq:penalty}
\begin{equation}\label{eq:alphadisug}
	\alpha_t(\bbQ) \geq k_t(\bbQ) \alpha^0_t(\bbQ^0) \ind_{t < \tau} + \alpha^1_t(\bbQ^1) \ind_{t \geq \tau},
\end{equation}
where $\bbQ^0$ and $\bbQ^1$ are probability measures on $(\Omega, \cF_T)$ and $(\Omega,\cH_T)$, respectively, such that
$$\dd \bbQ^0 = \bbE[L \mid \cF_T] \, \dd \bbP_{\vert \cF_T}, \quad \dd \bbQ^1 = L \dd \bbP_{\vert \cH_T}, \quad L \coloneqq \frac{\dd \bbQ}{\dd \bbP_{\vert \cG_T}},$$
$k_t(\bbQ)$ is an $\cF_t$-measurable random variable, depending on $\bbQ$ and satisfying $k_t(\bbQ) \geq 1$ $\bbP$-a.s., and
\begin{gather}
\alpha^0_t(\bbQ^0) \coloneqq \esssup_{\xi^0 \in \dL^\infty(\cF_T)} \bigl\{\bbE_{\bbQ^0}[-\xi^0 \mid \cF_t] - \rho^0_t(\xi^0) \bigr\}, \label{eq:alpha0}
\\
\alpha^1_t(\bbQ^1) \coloneqq \esssup_{\xi^1 \in \dL^\infty(\cH_T)} \bigl\{\bbE_{\bbQ^1}[-\xi^1(\tau, \zeta) \mid \cH_t] - \rho^1_t(\xi^1(\tau, \zeta)) \bigr\}. \label{eq:alpha1}
\end{gather}
Moreover, $\alpha_t(\bbQ) \ind_{t \geq \tau}= \alpha^1_t(\bbQ^1) \ind_{t \geq \tau}$.
\end{theorem}

\begin{remark}
Recall that, by (i) of~\citep[Prop. 2.7]{callegaro2013:carthaginian}, any $\cH_T$-measurable random variable $\xi^1$ is of the form $\xi^1 = x(\tau, \zeta)$, for some $\cF_T \otimes \cB(\R^+) \otimes \cB(E)$-measurable map $x$. To ease the notation, we will still use the symbol $\xi^1$ to denote the map $x$, as did in~\eqref{eq:alpha1}.
\end{remark}

\begin{proof}
Fix $t \in [0,T]$ and $\bbQ \in \cQ^{\hookrightarrow}$. By definition of $\alpha_t(\bbQ)$ and thanks to Propositions~\ref{prop:rhodecomp} and~\ref{prop:expectdecomp}, we have that
\begin{equation*}
	\alpha_t(\bbQ) \geq \bbE_\bbQ[-\xi \mid \cG_t] - \rho_t(\xi) = \bigl[\Phi^0 - \rho^0_t(\xi^0)\bigr] \ind_{t < \tau} + \bigl[\Phi^1(\tau,\zeta) - \rho^1_t(\xi^1(\tau,\zeta))\bigr] \ind_{t \geq \tau}
\end{equation*}
for any $\xi \in \dL^\infty(\cG_T)$. We concentrate, first, on the second part of the decomposition.

Since for any $\xi^1 \in \dL^\infty(\cH_T)$ we have that $\xi^1 \ind_{T \geq \tau} \in \dL^\infty(\cG_T)$, keeping in mind Remark~\ref{rem:Phi1}, we can write
\begin{eqnarray*}
\alpha_t(\bbQ) \ind_{t \geq \tau} & \geq & \bigl[\bbE_{\bbQ^1}[-\xi^1(\tau, \zeta) \ind_{T \geq \tau} \mid \cH_t] - \rho^1_t(\xi^1(\tau, \zeta) \ind_{T \geq \tau})\bigr] \ind_{t \geq \tau}
\\
&=& \bigl[\bbE_{\bbQ^1}[-\xi^1(\tau, \zeta) \ind_{T \geq \tau} \mid \cH_t] - \rho_t(\xi^1 \ind_{T \geq \tau})\bigr] \ind_{t \geq \tau}.
\end{eqnarray*}
Since $\{t \geq \tau\} \in \cG_t \subset \cH_t$, $\{t \geq \tau\} \subset \{T \geq \tau\}$ and recalling that $\rho$ satisfies the zero-one law by the assumptions listed at the beginning of this Section, we get
\begin{equation*}
\begin{aligned}
&\mathrel{\phantom{=}} \bigl[\bbE_{\bbQ^1}[-\xi^1(\tau, \zeta) \ind_{T \geq \tau} \mid \cH_t] - \rho_t(\xi^1 \ind_{T \geq \tau})\bigr] \ind_{t \geq \tau}
\\
&=
\bbE_{\bbQ^1}[-\xi^1(\tau, \zeta) \ind_{T \geq \tau} \ind_{t \geq \tau} \mid \cH_t] - \rho_t(\xi^1 \ind_{T \geq \tau} \ind_{t \geq \tau})
\\
&=
\bbE_{\bbQ^1}[-\xi^1(\tau, \zeta) \ind_{t \geq \tau} \mid \cH_t] - \rho_t(\xi^1 \ind_{t \geq \tau})
\\
&=
\bigl[\bbE_{\bbQ^1}[-\xi^1(\tau, \zeta) \mid \cH_t] - \rho_t(\xi^1)\bigr] \ind_{t \geq \tau}
=
\bigl[\bbE_{\bbQ^1}[-\xi^1(\tau, \zeta) \mid \cH_t] - \rho^1_t(\xi^1(\tau,\zeta))\bigr] \ind_{t \geq \tau}.
\end{aligned}
\end{equation*}
Therefore, $\alpha_t(\bbQ) \ind_{t \geq \tau} \geq \bigl[\bbE_{\bbQ^1}[-\xi^1(\tau, \zeta) \mid \cH_t] - \rho^1_t(\xi^1(\tau,\zeta))\bigr] \ind_{t \geq \tau}$
and taking the essential supremum on the r.h.s. with respect to all possible $\xi^1 \in \dL^\infty(\cH_T)$ we get
\begin{equation*}
\alpha_t(\bbQ) \ind_{t \geq \tau} \geq \alpha^1_t(\bbQ^1) \ind_{t \geq \tau}.
\end{equation*}
We can also get the reversed inequality, putting together Proposition~\ref{prop:expectdecomp}, the zero-one law property and the fact that $\dL^\infty(\cH_T) \supset \dL^\infty(\cG_T)$, as follows:
\begin{align*}
\alpha_t(\bbQ) \ind_{t \geq \tau}
&=
\esssup_{\xi \in \dL^\infty(\cG_T)} \bigl\{\bbE_\bbQ[-\xi \, | \, \cG_t] - \rho_t(\xi)\bigr\} \ind_{t \geq \tau}
\\
&= \esssup_{\xi \in \dL^\infty(\cG_T)} \bigl\{\bbE_\bbQ[-\xi \, | \, \cG_t] \ind_{t \geq \tau} - \rho_t(\xi) \ind_{t \geq \tau}\bigr\}
\\
&= \esssup_{\xi \in \dL^\infty(\cG_T)} \bigl\{\bbE_{\bbQ^1}[-\xi^1(\tau,\zeta) \, | \, \cH_t] \ind_{t \geq \tau} - \rho^1_t(\xi^1(\tau,\zeta)) \ind_{t \geq \tau}\bigr\}
\\
&= \esssup_{\xi \in \dL^\infty(\cG_T)} \bigl\{\bbE_{\bbQ^1}[-\xi^1(\tau,\zeta)  \ind_{t \geq \tau}\, | \, \cH_t] - \rho^1_t(\xi^1(\tau,\zeta) \ind_{t \geq \tau})\bigr\}
\\
&\leq \esssup_{\xi \in \dL^\infty(\cH_T)} \bigl\{\bbE_{\bbQ^1}[-\xi \ind_{t \geq \tau}\, | \, \cH_t] - \rho^1_t(\xi \ind_{t \geq \tau})\bigr\}
\\
&= \esssup_{\xi \in \dL^\infty(\cH_T)} \bigl\{\bbE_{\bbQ^1}[-\xi \, | \, \cH_t] - \rho^1_t(\xi)\bigr\} \ind_{t \geq \tau} = \alpha^1_t(\bbQ^1) \ind_{t \geq \tau}.
\end{align*}
Hence, we have that $\alpha_t(\bbQ) \ind_{t \geq \tau} = \alpha^1_t(\bbQ^1) \ind_{t \geq \tau}$.

We now turn our attention on the first part of the decomposition. Let us denote by $Z_t^\bbQ \coloneqq \bbQ(t < \tau \mid \cF_t)$ and $Z_t^\bbP \coloneqq \bbP(t < \tau \mid \cF_t)$ the $\bbQ$- and $\bbP$-Az\'ema supermartingales associated with $\tau$, respectively. For any $\xi^0 \in \dL^\infty(\cF_T)$, having in mind the definition of $\Phi^0$ given in~\eqref{eq:Phi01}, we get
\begin{align*}
&\mathrel{\phantom{\geq}} \alpha_t(\bbQ) \ind_{t < \tau}
\\
&\geq \biggl(\frac{1}{Z_t^\bbP} \bbE\biggl[\int_0^{+\infty} \int_E -\xi^0 \bigl\{ L^0 \ind_{(T, +\infty)}(\theta) + L^1(\theta,e) \ind_{(t,T]}(\theta)\bigr\} \gamma_T(\theta, e) \, \dd \theta \, \dd e \mid \cF_t\biggr] \biggr.
\\
&\mathrel{\phantom{\geq}} \qquad \biggl. - \rho^0_t(\xi^0)\biggr) \ind_{t < \tau}
\\
&= \biggl(\frac{1}{Z_t^\bbP} \bbE\bigl[ \bbE[-\xi^0 \{L^0 \ind_{T < \tau} + L^1(\tau,\zeta) \ind_{t < \tau} \ind_{T \geq \tau}\} \mid \cF_T] \mid \cF_t \bigr] - \rho^0_t(\xi^0)\biggr) \ind_{t < \tau}
\\
&= \biggl(\frac{1}{Z_t^\bbP} \bbE\bigl[-\xi^0 \bbE[\ind_{t < \tau} \{L^0 \ind_{T < \tau} + L^1(\tau,\zeta) \ind_{T \geq \tau}\} \mid \cF_T] \mid \cF_t \bigr] - \rho^0_t(\xi^0)\biggr) \ind_{t < \tau}
\\
&= \biggl(\frac{1}{Z_t^\bbP} \bbE\bigl[-\xi^0 \bbE[L \ind_{t < \tau} \mid \cF_T] \mid \cF_t \bigr] - \rho^0_t(\xi^0)\biggr) \ind_{t < \tau}
\\
&= \biggl(\bbE\biggl[-\xi^0 \frac{\bbE[L \ind_{t < \tau} \mid \cF_T]}{Z_t^\bbP} \biggm| \cF_t \biggr] - \rho^0_t(\xi^0)\biggr) \ind_{t < \tau}.
\end{align*}
where we used the fact that $\{T < \tau\} \subset \{t < \tau\}$.

Notice that $\bbE[L \mid \cF_T] > 0$, since $L > 0$, therefore using the conditional Bayes' formula and \citep[Lemma 3.8]{aksamitjeanblanc2017:filtrenl} (where the immersion property is essential), we get
\begin{equation*}
\bbE[L \ind_{t < \tau} \mid \cF_T] = \bbQ(t < \tau \mid \cF_T) \, \bbE[L \mid \cF_T] = \bbQ(t < \tau \mid \cF_t) \, \bbE[L \mid \cF_T],
\end{equation*}
hence
\begin{equation*}
\alpha_t(\bbQ) \ind_{t < \tau} \geq \biggl(\frac{Z_t^\bbQ}{Z_t^\bbP} \bbE\bigl[-\xi^0 \bbE[L \mid \cF_T] \mid \cF_t \bigr] - \rho^0_t(\xi^0)\biggr) \ind_{t < \tau} = \biggl(\frac{Z_t^\bbQ}{Z_t^\bbP} \bbE_{\bbQ^0}\bigl[-\xi^0 \!\mid\! \cF_t \bigr] - \rho^0_t(\xi^0)\biggr) \ind_{t < \tau}.
\end{equation*}

From now on, let us restrict our choice for $\xi^0$ to the positive ones. Otherwise said, since $\rho_t$ is normalized thanks to the zero-one law, we only consider $\xi^0$ such that $\rho_t(\xi^0) \leq 0$. Then, easy computations show that
\begin{equation*}
\alpha_t(\bbQ) \ind_{t < \tau} \geq \biggl[\biggl(\frac{Z_t^\bbQ}{Z_t^\bbP} - 1\biggr) \ind_{\Bigl\{\frac{Z_t^\bbQ}{Z_t^\bbP} \geq 1\Bigr\}} + 1\biggr] \bigl(\bbE_{\bbQ^0}\bigl[-\xi^0 \mid \cF_t \bigr] - \rho^0_t(\xi^0)\bigr) \ind_{t < \tau}.
\end{equation*}
Since the term in square brackets is positive and $\rho_t(\xi^0) \leq 0$, we get
\begin{equation*}
\alpha_t(\bbQ) \ind_{t < \tau} \geq \biggl[\biggl(\frac{Z_t^\bbQ}{Z_t^\bbP} - 1\biggr) \ind_{\Bigl\{\frac{Z_t^\bbQ}{Z_t^\bbP} \geq 1\Bigr\}} + 1\biggr] \bbE_{\bbQ^0}\bigl[-\xi^0 \mid \cF_t \bigr] \ind_{t < \tau}
\end{equation*}
and taking the essential supremum with respect to all $\xi^0 \geq 0$, i.e., w.r.t. all $\xi^0 \in \dL^\infty(\cF_T)$ such that $\rho_t(\xi^0) \leq 0$, we have
\begin{equation*}
\alpha_t(\bbQ) \ind_{t < \tau} \geq \biggl[\biggl(\frac{Z_t^\bbQ}{Z_t^\bbP} - 1\biggr) \ind_{\Bigl\{\frac{Z_t^\bbQ}{Z_t^\bbP} \geq 1\Bigr\}} + 1\biggr] \alpha^0_t(\bbQ^0) \ind_{t < \tau} = k_t(\bbQ) \alpha^0_t(\bbQ^0) \ind_{t < \tau}.
\end{equation*}
\end{proof}

\begin{remark}\label{rem:Lrefmeas}
It is worth noticing that a similar conclusion holds for all $\bbQ \in \cQ$ such that the Radon-Nikod\'ym density $L$ is $\cF_T$-measurable. In this case the immersion property is automatically satisfied, i.e., $\bbQ \in \cQ^{\hookrightarrow}$, and moreover $k_t(\bbQ) = 1$ for all $t \in [0,T]$ (see \citep[Prop. 3.6 (c)]{aksamitjeanblanc2017:filtrenl}. To be precise, a careful inspection of the previous proof reveals that it is not necessary, in this case, to invoke the immersion property at all).
\end{remark}

\begin{remark}
It should be clear from Proposition~\ref{prop:alphadecomp} and Remark~\ref{rem:Lrefmeas} that, without any additional structure on the set $\cQ$ of equivalent probability measures with respect to $\bbP_{\vert \cG_T}$, there is no hope of reaching a decomposition of the penalty term $\alpha_t(\cdot)$ similar to that provided for the dynamic risk measure $\rho$ in~\eqref{eq:rhodecomp}.

In fact, while $\alpha_t(\bbQ)$ reduces to $\alpha_t ^1(\bbQ^1)$ on the event $\{t \geq \tau \}$ since $\bbG$ and $\bbH$ coincide, on $\{t < \tau \}$ it exceeds $\alpha_t ^0(\bbQ^0)$ multiplied by a dilation factor greater than or equal to $1$ and depending on $\bbQ$ and $\tau$. From a mathematical point of view, this is due to the fact that we do not know anything \emph{a priori} about how the equivalent probability measure $\bbQ$ modifies the probability of the default-related event $\{t < \tau\}$, i.e., how it modifies the $(\bbP, \bbF)$-Az\'ema supermartingale $\bbP(t < \tau \mid \cF_t)$.

Also from a financial point of view, this fact (hence also inequality~\eqref{eq:alphadisug}) is quite reasonable and can be interpreted as follows. Since the higher is the penalty associated to a scenario the lower is the confidence one has on that scenario, the inequality on $\alpha_t(\bbQ)$ and, in particular, the presence of the factor $k_t(\bbQ) \geq 1$ can be seen as a further penalization on the part of the penalty term before $\tau$ in order to take into account the lack of information. On the contrary, the penalty term $\alpha_t(\bbQ)$ reduces to $\alpha_t ^1(\bbQ^1)$ once default is occurred.
\end{remark}

\medskip

{\color{blue}\noindent \textbf{Acknowledgments.} The authors wish to thank Claudio Fontana for pointing out~\citep{song2014:optionalsplitting} and for a useful discussion on the optional splitting formula.}

\bibliographystyle{plainnat}
\bibliography{Bibliography}

\end{document}